\documentclass[letter,11pt]{article}
\pdfoutput=1
\usepackage{amssymb, amsthm}
\usepackage[fleqn]{amsmath}
\usepackage{complexity}
\usepackage{float}
\usepackage{algorithm}
\floatname{algorithm}{Protocol}
\usepackage{multirow}
\usepackage{graphicx}
\usepackage[margin=1in]{geometry}
\usepackage[normalem]{ulem}

\usepackage{color}

\newcounter{linenum}

\newcommand{\bigabs}[1]{{\big\lvert #1\big\rvert}}

\newcommand{\ket}[1]{\left\lvert #1 \right\rangle}
\newcommand{\bra}[1]{\left\langle #1 \right\lvert}

\DeclareMathOperator{\Tr}{\operatorname{Tr}}

\providecommand{\abs}[1]{\lvert#1\rvert}

\def\XA {\mathcal{X}^A}
\def\xzdeterminedset {{\mathcal Q}}
\def\H {{\mathcal H}}

\usepackage{color}

\newtheorem{theorem}{Theorem}
\newtheorem{lemma}{Lemma}
\newtheorem{corollary}{Corollary}

\newtheorem{definition}{Definition}

\usepackage[affil-it]{authblk}

\usepackage{hyperref}
\usepackage{natbib}
\usepackage{braket}
\usepackage{xspace}

\newcommand{\QPIP}{\textsf{QPIP}\xspace}
\newcommand{\MIPk}[1]{\textsf{MIP[#1]}\xspace}
\newcommand{\MIPs}{\textsf{MIP\large{*}}\xspace}

\begin{document}

\allowdisplaybreaks[3]
\frenchspacing

\title{Robustness and device independence of \\verifiable blind quantum computing}

\author[1]{Alexandru Gheorghiu\footnote{a.gheorghiu@sms.ed.ac.uk}}
\author[1,2]{Elham Kashefi}
\author[1]{Petros Wallden}
\affil[1]{School of Informatics, University of Edinburgh,}
\affil[ ]{10 Crichton Street, Edinburgh EH8 9AB, UK}
\affil[2]{CNRS LTCI, Departement Informatique et Reseaux,}
\affil[ ]{Telecom ParisTech, Paris CEDEX 13, France}

\hyphenation{avenues}
\hyphenation{research}
\hyphenation{interact}
\hyphenation{machine}
\hyphenation{machines}
\hyphenation{giving}
\hyphenation{paper}
\hyphenation{encoun-ter}
\hyphenation{encoun-ters}
\hyphenation{expe-ri-ence}
\hyphenation{analyse}
\hyphenation{analysis}
\hyphenation{remains}
\hyphenation{logical}
\hyphenation{cannot}

\clubpenalty 10000
\widowpenalty 10000

\date{}

\maketitle

\begin{abstract}
Recent advances in theoretical and experimental quantum computing bring us closer to scalable quantum computing devices. This makes the need for protocols that verify the correct functionality of quantum operations timely and has led to the field of quantum verification. In this paper we address key challenges to make quantum verification protocols applicable to experimental implementations. We prove the robustness of the single server verifiable universal blind quantum computing protocol of Fitzsimons and Kashefi \cite{fk} in the most general scenario. This includes the case where the purification of the deviated 
input state is in the hands of an adversarial server. The proved robustness property 
allows
the composition of this protocol with a device-independent state tomography protocol that we give, which is based on the rigidity of CHSH games as proposed by Reichardt, Unger and Vazirani~\cite{ruv2}. The resulting composite protocol has lower round complexity for the verification of entangled quantum servers with 
a classical verifier and, as we show, can 
be made fault tolerant.
\end{abstract}

\vspace{2pc}

\section{Introduction}
While the prospect of commercially available universal quantum computing is still distant, a number of experiments involving multi-qubit systems have recently been developed.
Irrespective of their applications, these technologies require methods and tools for verifying the correctness of their operations. Assuming that quantum computing is more powerful than classical computing, a simulation-based approach for quantum verification of devices with sufficiently large number of qubits, becomes practically impossible. Aaronson and Arkhipov showed in \cite{aaronson} that even a rudimentary quantum computer constructed with linear-optical elements cannot be efficiently simulated. Similarly, verifying the correct preparation of a general $n$ qubit state via state tomography also involves exponential overhead since it requires collecting statistics from $4^n$ separate observables \cite{nielsenchuang}.

The verification of quantum devices becomes more complicated when the functionality involves cryptographic primitives. In these cases, incorrect operations could be the result of actions of an adversary. Thus it becomes necessary to guarantee the security of the application under certain assumptions about the devices.
Ideally a protocol should remain secure even if the devices are faulty and partially controlled by adversaries. This would lead to a solution that is \emph{device-independent} and \emph{robust}. However, generating such protocols has
proven difficult. Even in quantum key distribution, a complete proof of security for a device-independent protocol,
in the presence of noise, has been achieved only
recently
\cite{deviceindependentqkd}.

The issue of verification needs to be resolved to be able to exploit successfully any future quantum computers. 
Moreover, one expects that the first large scale quantum devices are unlikely to be personal computers. Instead, they will probably function as \emph{servers} to which \emph{clients} can connect and request the computation of some difficult problem. The client may also require his computation to be private, i.e. require that the server does not learn anything about it. We should therefore construct protocols that verify an arbitrary delegated quantum computation and prove the security and correctness of this verification technique.

The approaches that have been so far successful are those based on \emph{interactive proof systems} \cite{ip, qip}, where a \emph{trusted}, computationally limited verifier (also known as client, in a cryptographic setting) exchanges messages with an \emph{untrusted}, powerful quantum prover, or multiple provers (also known as servers). The verifier attempts to certify that, with high probability, the provers are performing the correct quantum operations. Because we are dealing with a new form of computation, the verification protocols, while based on established techniques are fundamentally different from their classical counterparts. A number of quantum verification protocols have been developed,
for different functionalities of devices and using a variety of different strategies to achieve verification \cite{fk,ruv2,bfk,eleni, abe, bosonsampling, mckague,KKD14, twopartyeval, cryptoquantumw, BGS13, efk}. The assumptions made depend  on the specific target and desired properties of the protocol. For example, if
the emphasis is on creating an immediate practical implementation,
then this should be reflected in the technological requirements leading to a testable application with current technology \cite{efk}.
Alternatively, if the motivation is to prove a theoretical result, we may relax some requirements such as efficient scaling \cite{ruv2}. An important open problem in the field of quantum verification, is whether a scheme with a fully classical verifier is possible \cite{falsifiable, openproblem}. We know, however, that verification \emph{is} possible in the following two scenarios:
\begin{enumerate}
\item A verifier with minimal quantum capacity (ability to prepare random single qubits) and a single quantum prover \cite{fk}. This is the Fitzsimons and Kashefi (FK) protocol.
\item A fully classical verifier and two non-communicating quantum provers that share entanglement \cite{ruv}. This is the Reichardt, Unger and Vazirani (RUV) protocol.
\end{enumerate}

One of our objectives is to obtain a \emph{device independent} (allowing untrusted quantum devices) version of the FK protocol, by composing it with the RUV protocol.
The additional properties we aim to achieve from this composition are \emph{fault tolerance} (allowing noisy devices) and \emph{reduced round complexity}.
Composing protocols can indeed be fruitful since it could lead to new protocols that inherit the advantages of both constituents. 
The universal composabillity framework, allowing for secure compositions, has been successfully extended to the quantum realm \cite{compos1, compos2, compos3}. Recently, the security of single server verifiable universal blind quantum computing protocols has been demonstrated in an abstract cryptographic framework \cite{DFPR13} that is also known to be equivalent to the simulation-based composability framework.
However this setting does not fulfil the necessary requirements for our composition. 
This is because, when combining a single server verification scheme (FK) with an entangled server scheme (RUV), there exists the possibility of correlated attacks, which are not explicitly treated in the composability framework. Such attacks can occur when an untrusted server's strategy is correlated with deviations in the protocol's input state. Our robustness result resolves this issue in the stand alone composition setting and the same technique could potentially be extended to the composable framework of \cite{DFPR13}, thus resolving the problem of correlated attacks.

The type of composition that we require is sequential.
We take the output of the first protocol and use it as input for the second. However, in general, the output of the first protocol is not necessarily an acceptable input for the second protocol. In particular, since the verification protocols are probabilistic, their outputs typically deviate by a small amount from the ideal one. Thus, it is necessary that the second protocol remain secure even if the input is slightly deviated from the ideal one. Moreover, we make sure that adversaries cannot exploit any correlations between the deviated input and their strategy to compromise the security of the protocol. Therefore to securely compose the protocols, we need to address these new type of attacks.
The main results of this paper can be summarised as follows:
\begin{enumerate}
\item We prove that the FK protocol is strongly robust, see Theorem \ref{t-robust}. First, we show that FK can tolerate inputs which deviate from their ideal values by a small amount, see Lemma \ref{l-weakrobust}. However, for composition with other protocols a stronger property is needed. We therefore proceed to show that the FK protocol is robust even when the deviated input is correlated with an external system possessed by an adversary (for example the provers' private systems in the RUV protocol), see Lemma \ref{l-strongrobust}.

\item An immediate consequence of the robustness theorem is that we can construct a composite protocol combining RUV with FK. The required input states for the FK protocol are prepared via the state tomography sub-protocol of RUV. Our composite protocol inherits the device independence property of RUV, see Theorem \ref{t:composite}. Additionally, since we do not require the full RUV protocol, the composite protocol also has an improved round complexity.

\item Lastly, we address the distinction between robustness and fault tolerance and show how the FK protocol can be made fault tolerant,
thereby making our proposed composite approach fault tolerant as well.
\end{enumerate}
In Section~\ref{sect:pre} we give some preliminaries. In Section~\ref{sect:main} we present the main results, that we summarised above, and outline their proofs. In particular we give robustness in Subsection~\ref{sect:main-robust}, composition in Subsection~\ref{sect:main-comp} and fault tolerance in Subsection~\ref{sect:main-ft}.  Further details of the proofs are given in Section~\ref{sect:robust} for robustness, in Section~\ref{sect:composite} for composition and in Section~\ref{sect:ft} for fault tolerance. We conclude in Section~\ref{sect:conclusion}.

\subsection{Preliminaries}\label{sect:pre}
We first introduce the relevant concepts used in describing verification protocols and then briefly present the
two protocols we will built on (FK and RUV).

\subsubsection{Interactive Proof Systems}
As explained in \cite{ip, abe}, a language $\mathcal{L}$ is said to admit an interactive proof system if there exists a computationally unbounded
prover $\mathcal{P}$ and a \BPP{}  verifier $\mathcal{V}$, such that for any $x \in \mathcal{L}$, $\mathcal{P}$ convinces $\mathcal{V}$
that $x \in \mathcal{L}$ with probability $\geq \frac{2}{3}$. Additionally, when $x \not\in \mathcal{L}$, $\mathcal{P}$ convinces $\mathcal{V}$
that $x \in \mathcal{L}$ with probability $\leq \frac{1}{3}$. Mathematically, we have the following two conditions\footnote{
Note that completeness can be viewed as the probability of the verifier accepting when the prover is honest. Similarly, soundness
is the probability of accepting, when the prover is dishonest.}:
\begin{itemize}
\item Completeness: $Pr( \mathcal{V} \leftrightarrow \mathcal{P}$ accepts $ x \; \; | \; \; x \in \mathcal{L}) \geq \frac{2}{3}$
\item Soundness: $Pr( \mathcal{V} \leftrightarrow \mathcal{P}$ accepts $ x \; \; | \; \; x \not\in \mathcal{L}) \leq \frac{1}{3}$
\end{itemize}
The set of languages which admit such an interactive proof system define the complexity class \textsf{IP}. We are interested in the case
when the prover is a polynomial-time quantum computer (i.e. a \BQP{} machine).
In \cite{abe}, the first definition of such a quantum interactive proof system was given, which we use here:
\begin{definition} \label{def:qpip}
\cite{abe} Quantum Prover Interactive Proof (\textsf{QPIP}) is an interactive proof system with the following properties:
\begin{itemize}
\item[(i)] The prover is computationally restricted to \BQP{}.
\item[(ii)] The verifier is a hybrid quantum-classical machine. Its classical part is a \BPP{} machine. The quantum part is a
register of c qubits (for some constant c), on which the prover can perform arbitrary quantum operations. At any
given time, the verifier is not allowed to possess more than c qubits. The interaction between the quantum and
classical parts is the usual one: the classical part controls which operations are to be performed on the quantum
register, and outcomes of measurements of the quantum register can be used as input to the classical machine.
\item[(iii)] There are two communication channels: one quantum and one classical.
\end{itemize}
The completeness and soundness conditions are identical to the \textsf{IP} conditions.
\end{definition}

We are also interested in interactive protocols that use more than one prover.
There are only two differences, first that the verifier can interact with multiple provers instead of just one, and second that the provers are not allowed to
communicate. The conditions for completeness and soundness
remain unchanged.
The analogous complexity class that involves multiple provers is called \emph{Multi-Prover Interactive Proof System} and denoted \MIP{} \cite{mip}.
It is defined as the set of all languages which admit an
interactive proof system with one or more non-communicating provers. If the number of provers is fixed to be $k$, the corresponding complexity class is
\MIPk{$k$}. A closely related class is \MIPs where the multiple non-communicating provers share entangled states.

In all of these cases, the verifier is essentially \emph{delegating} a difficult computation to the prover(s). This computation can be universal
with respect to the computation model of the prover(s). In our case, this means universal for polynomial-time quantum computations.
The number of classical messages exchanged between the verifier and the prover, throughout the run of the protocol, as a function of the input size
is known as the \emph{round complexity} of the protocol.

\subsubsection{Quantum Protocols}
Throughout this subsection, we assume the reader is familiar with the teleportation-based and more generally measurement-based quantum computing (MBQC) models, described in detail in \cite{onewaycomputer, mbqc}.

We first summarise the FK protocol \cite{fk} which is a \QPIP protocol.  It is also known as unconditionally secure,
verifiable, universal blind quantum computing. The protocol is ``blind'' which means that no information about the computation is leaked to the prover, apart from its size. This property can be exploited by allowing the verifier to insert hidden ``traps'' within the computation. 
The traps are deterministic tests which the verifier can perform in order to verify that the prover is not deviating from the protocol. Blindness ensures that the traps are indistinguishable from the computation.

The basic idea of this protocol is that the verifier prepares and sends qubits to the prover. The prover entangles these qubits and then performs adaptive measurements (sending the measurement outcomes to the verifier) that will overall implement a certain unitary operation, as in the MBQC model of computation.
The traps are single isolated qubits, disentangled from the rest of the computation, and when measured in suitable bases give deterministic outcomes that are known to the verifier (but not to the prover). Since the prover is completely blind and does not know which qubits are traps  and which are part of the actual computation, any attempt to cheat has some probability to affect the trap and thus be detected.  

The FK protocol is based on a universal resource state for the MBQC model, known as the \emph{dotted-complete graph state}.
The details of this resource state are not crucial for understanding this paper, apart from the fact that, as part of the FK protocol, the appropriate operators are performed by the untrusted server to prepare this generic state. In particular, a series of controlled-$Z$ operators are performed by the server, according to the dotted-complete graph structure, for entangling the individual qubits prepared in advance by the verifier. These initial qubits, that are sometimes referred to as the input of the FK protocol, are sent to the server at the first stage of the protocol. This fact is used to prove some basic properties needed for our main robustness result, see Theorem \ref{thm:stabilizer}. Therefore, for the purpose of completeness, we state here the definition of the dotted-complete graph state, taken from \cite{fk}, see also Figure~\ref{f-dgrapg}.

\begin{definition} \label{d-dotted}
\cite{fk}
Let $K_N$ denote the complete graph of $N$ vertices. Define the \emph{dotted-complete graph}, denoted as $\tilde{K}_N$, to be a graph where every edge in $K_N$ is replaced with a new vertex connected to the two vertices originally joined by that edge. We call the quantum state corresponding to $\tilde{K}_N$ the \emph{dotted-complete graph state}. This multi-partite entangled state is prepared by replacing every vertex with a qubit in the state $\ket +$ and applying a controlled-$Z$ operator for every edge in the graph.
\end{definition}

\begin{figure}[h!]
  \centering
  \includegraphics[width=0.8\textwidth]{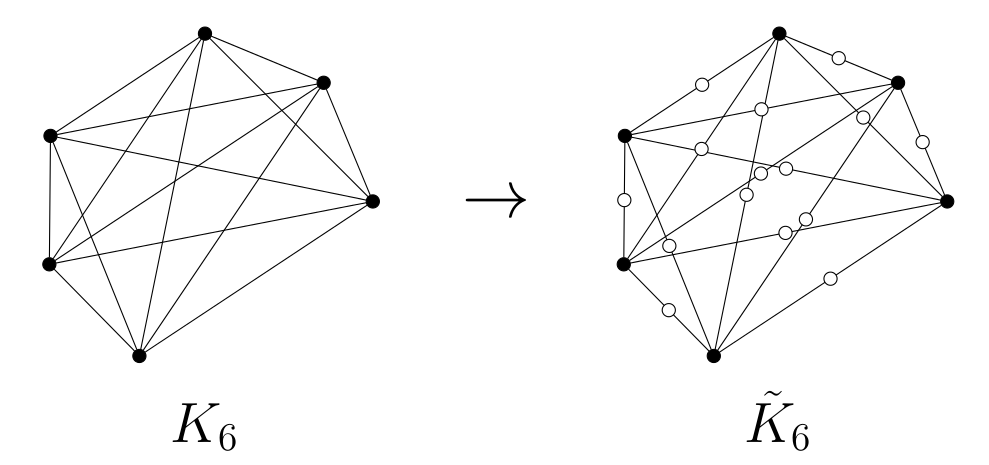}
  \caption{An example of a complete graph, $K_6$, and its corresponding dotted-complete graph $\tilde{K}_6$.}
\label{f-dgrapg}
\end{figure}

The family of dotted-complete graph states is universal for quantum computation. Moreover, any other graph state could be obtained from a large enough dotted-complete graph state by applying the appropriate Pauli measurements over some of the vertices (the ones shown in white, in Figure~\ref{f-dgrapg}). 
Concretely, in order to construct any desired graph of $N$ vertices from a dotted-complete graph $\tilde{K}_N$, Pauli $Y$ measurements are performed in order to keep a specific edge, and Pauli $Z$ measurements in order to remove it (alternatively, one could use the states $\Ket{0}$, $\Ket{1}$ for the edges which should be removed, instead of performing a Pauli $Z$ measurement). This can be done blindly in order to hide the target graph.
The detailed construction is not important for the rest of this paper and hence omitted (see Section 5 in \cite{fk}). We give a brief description of FK, shown here as Protocol~\ref{prot:fk}. \\
According to Definition~\ref{def:qpip} the quantum channel between verifier and prover is one-way (from the verifier to the prover). Moreover, the constant $c$, representing the number of qubits that the verifier can possess at any given time, is exactly one. We refer the reader to \cite{fk} for a more in depth description of the protocol and its associated concepts. However, we recall the key properties of the protocol in the following lemma.

\begin{algorithm}[t!]
\caption{Fitzsimons, Kashefi \QPIP Protocol from \cite{fk}}
 \label{prot:fk}
 \vskip 0.2 cm
\textbf{Assumptions} \\
The verifier wants to delegate a quantum computation described by the graph $G$ and specific measurement angles $\phi_i$ $\in \{0, \pi/4 ... 7\pi/4 \}$ chosen to define a desired computation. He chooses a \emph{security parameter} $d$ and
encodes this graph into the topological error correcting scheme of Raussendorf, Harrington and Goyal \cite{rhg},
that can correct or detect errors with weight less than $d$.
The encoded graph, $\mathcal{G}$, will have $N$ qubits.
He then considers a random hidden partitioning of the vertices of the dotted-complete graph $\widetilde{\mathcal{K}}_{3N}$ into $3$ subgraphs: the computation graph $\mathcal{G}$, and two different types of isolated traps which the verifier will use in order to test the prover's honesty in performing measurements. All measurements are performed in the $XY$-plane.\\

\textbf{Protocol}
\begin{enumerate}
\item Verifier prepares and sends the $3N(3N + 1)/2$ qubits to the prover (the number of vertices in $\widetilde{\mathcal{K}}_{3N}$). These consist of dummy qubits which are either $\Ket{0}$ or $\Ket{1}$
(isolating their neighbouring vertices from the rest of the computation) and computation or trap qubits which are of the form $\Ket{+_{\theta}} = \frac{1}{\sqrt{2}}(\Ket{0} + e^{i\theta} \Ket{1})$,
where $\theta \in \{0, \pi/4 ... 7\pi/4 \}$.
\item The prover entangles the qubits according to the structure of the $\widetilde{\mathcal{K}}_{3N}$ graph by applying controlled-Z operations between any pair of vertices that are connected with an edge.
\item For each qubit $i$, the verifier
computes the measurement angle $\delta_i = \theta_i + \phi'_i + r_i \pi$, where $\phi'_i $ is the adapted version of the computation angle $\phi_i$, and $r_i$ is a randomly chosen bit $\{0,1\}$. Adapted computation angles are used to account for corrections from previous measurements. The measurement angles $\phi_i$, for the trap qubits, are randomly fixed to be $0$ or $\pi$. However, due to blindness provided by the initial $\theta_i$ rotations in the preparation of individual qubits (Step 1 above), the value of $\delta_i$  is uniformly distributed over the set $\{0, \pi/4 ... 7\pi/4 \}$. The verifier sends these measurement angles one by one to the prover. The prover measures each corresponding qubit in the $\Ket{+_{\delta_i}}$, $\Ket{-_{\delta_i}}$ basis, and sends his reply $b_i$ to the verifier.
\item The verifier accepts if for all trap qubits, the reported measurement outcome $b_t$ is the same as the expected outcome $r_t$.
\end{enumerate}
\end{algorithm}

\begin{lemma}
Assuming the verifier wants to delegate the computation of a circuit of size $N$, the FK protocol has $O(N^2)$ round complexity and uses
$O(N^2)$ qubits with completeness being exactly 1 while the soundness is upper bounded by $(2/3)^{\lceil \frac{2d}{5} \rceil}$, where $d$ is the security parameter.

\end{lemma}
\begin{proof}
It is clear from Protocol~\ref{prot:fk} that the total number of qubits used in the protocol is $3N(3N + 1)/2$,
where $N$ is the number of qubits for the encoded graph $\mathcal{G}$.
Additionally, we have the same number of rounds of classical communication, corresponding to the measurements of qubits in the dotted-complete graph
(each measurement requires $3$ classical bits to specify the measurement angle and $1$ bit to specify the outcome).
Both the overall round complexity and the required quantum resources are thus $O(N^2)$.

As described in Protocol~\ref{prot:fk}, the verifier accepts if and only if all trap measurements succeed. This is always the case if the prover is honest and follows the instructions, since the trap measurements are deterministic. Therefore, the probability that the verifier accepts when the prover is honest is exactly $1$ (completeness). On the other hand, it is shown in \cite{fk} that in the case of classical output, the protocol is $(2/3)^{\lceil \frac{2d}{5} \rceil}$-verifiable, meaning the soundness is upper bounded by $(2/3)^{\lceil \frac{2d}{5} \rceil}$.
\end{proof}

Next we summarise the RUV protocol \cite{ruv2} which is a \MIPs protocol. It relies on the rigidity of CHSH games \cite{CHSH69, ruv} to test the honesty of the provers (c.f. traps at the FK protocol),
and on
gate teleportation to perform the computation. In particular,
the verifier directs the provers to perform a series of local measurements of their parts of the shared entangled states. The purpose of this is to check for statistical violations of Bell's inequality. At the same time, the verifier makes the provers teleport quantum states into gates in order
to perform his desired quantum computation \cite{ruv}. Importantly, the verifier alternates between these strategies in such a way that the provers are not aware (they are blind) in which strategy their measurement belongs. Moreover, the two provers cannot use previous results in order to deviate from the protocol (i.e. there is no adaptive cheating strategy). This is summarised as Protocol~\ref{prot:ruv}.
Due to the rigidity of the CHSH games, as proved in \cite{ruv2}, the verifier can determine if the two provers are being honest or not from the statistical outcomes. To ensure the verification of universal computations, the resource preparation stage of the protocol will prepare multiple copies of the states:
\begin{equation*} \label{eqn:state}
\big\{
   P \ket 0 , \; (H P)_2 \ket{\psi^*}, \; (G Y)_2 \ket{\psi^*}, \; \mathrm{CNOT}_{2,4} P_2 Q_4 (\ket{\psi^*} \otimes \ket{\psi^*}) \, : \, P, Q \in \{I, X, Y, Z\}
\big\}
\end{equation*}
Here, $\Ket{\psi^*}$ denotes the Bell state $\frac{\Ket{00} + \Ket{11}}{\sqrt{2}}$ which is shared among the two provers. In fact the provers share multiple copies of $\Ket{\psi^*}$, each prover having one qubit from each Bell pair. Without loss of generality,
we can assume that prover $1$ has the first qubit and prover $2$ has the second qubit.
The Hadamard, Phase and controlled-Not gates (denoted as $ \{ H, G, CNOT \} $) constitute a universal gate set for quantum computation. The subscript indices indicate on which qubits the gate acts. An arbitrary quantum circuit is thus simulated by repeatedly doing
gate teleportations, while keeping the computation blind from the two provers the entire time.

\begin{lemma}\label{l:ruv-complexity}
Assuming the verifier wants to delegate the computation of a circuit of size $n$, the round complexity of the RUV protocol is $O(n^c)$,
where there exists some constant $c$, such that $c \geq 8192$.
\end{lemma}
\begin{proof}
To determine an upper bound for the round complexity we only need to inspect the number of rounds of CHSH games, since the protocol randomly alternates
between this and the other three subprotocols. As shown in Protocol~\ref{prot:ruv},
the verifier plays $N$ sets of CHSH games with the two provers. Each
set consists of $qn_s$ games and $n_s \geq n^{64}$.
Additionally, it is required that $N \geq (qn_s)^{\alpha - 1}$, where $n^{\alpha/2} \geq n^{64}$, so $\alpha \geq 128$.
These conditions are necessary for the correctness of the state tomography and process tomography subprotocols \cite{ruv}.
We then have that $N \geq d n^{8128}$, where $d$ is a constant of the form $d = q^{\alpha - 1}$.
This is the number of sets of CHSH games and hence the number of required games is lower bounded by $n^{8192}$. It follows that the number of rounds is $O(n^c)$, where $c \geq 8192$. Note that by ``lower bounded'' we refer to the case when all of the CHSH statistics are consistent and the verifier does not reject. In the case of inconsistent statistics, the verifier can reject before playing $n^{8192}$ games.
\end{proof}

The subprotocols of the RUV protocol are themselves verification protocols as proved in \cite{ruv}. The subprotocol that we will use is the state tomography protocol.
As part of the RUV protocol, it is used to prepare resource states which are \emph{XZ-determined}, i.e. states that are uniquely determined
by their traces against $X$ and $Z$ operators. To compose the RUV with the FK protocol we will use a modified version of the state tomography subprotocol, so that we can prepare all states that are allowed inputs for the FK protocol.
We will give the modified protocol in the next section.

\begin{algorithm}[h!]
\caption{Reichardt, Unger, Vazirani \MIPs Protocol from \cite{ruv} (for two provers)}
 \label{prot:ruv}
 \vskip 0.2 cm
\textbf{Assumptions} \\
The verifier delegates a quantum circuit of size $n$ to two quantum provers.
Let $n_s = n^{\alpha/2} \geq n^{64}$, $q = 11$, $n_g = qn_s$, $N \geq n_g^{\alpha - 1}$ and $\delta = 1 / (6n^{\alpha/8})$.
The two provers share $Nn_g$ Bell states. \\

\textbf{Protocol} \\
The verifier alternates randomly between four subprotocols. He chooses the first three with probability $(1 - \delta)/3)$ and the last one
with probability $\delta$.
\begin{enumerate}
\item \textbf{CHSH games.} The verifier referees $N$ sets of sequential CHSH games, each consisting of $n_g$ games between the provers.
He rejects if they win less than:
\[
\cos^2(\pi/8) Nn_g - \frac{1}{2\sqrt{2}}\sqrt{N n_g \log(N n_g)}
\]
of the games.
\item \textbf{State tomography.} The verifier chooses $K \in [N]$ uniformly at random and referees $K - 1$ sets of CHSH games. He sends the questions
from the $K$th set to prover $1$, while running a state tomography protocol with prover $2$.
In this protocol prover $2$ is asked to prepare $q$-qubit resource states by
measuring his halves of the shared Bell states. This will collapse prover $1$'s states to the same
$q$-qubit resource states up to corrections. The verifier checks
this using the CHSH measurement outcomes from prover $1$. These outcomes tomographically determine the states
that are being prepared. He rejects if the tomography statistics are inconsistent.
\item \textbf{Process tomography.} The verifier
chooses $K \in [N]$ uniformly at random and referees $K - 1$ sets of CHSH games. He sends the questions
from the $K$th set to prover $2$,
while running a process tomography protocol with prover $1$. In this protocol prover $1$ is asked to perform Bell measurements on
his halves of the shared Bell states. The verifier checks this using the CHSH measurement outcomes from prover $2$.
He rejects if the tomography statistics are inconsistent.
\item \textbf{Computation.} The verifier
chooses $K \in [N]$ uniformly at random and refereed $K - 1$ sets of CHSH games. In the $K$th game he runs a state
tomography protocol with prover $2$ and a process tomography protocol with prover $1$. The combination of these two achieves computation via
gate teleportation.
\end{enumerate}
\end{algorithm}

\section{Main Results}\label{sect:main}

\subsection{Robustness}\label{sect:main-robust}

The first result we prove is that the FK protocol is robust with respect to small variations in the input. Throughout this paper, by ``input'' we are referring to the quantum states that the verifier sends to the prover and not the computation input. Without loss of generality we can assume that the desired computation that will be delegated to the server has the fixed classical input $0,\ldots 0$. Dealing with arbitrary classical or quantum input is straightforward, as explained in \cite{fk}, and makes no difference for our result. Hence, for the rest of this paper we define the \emph{input state} of the FK protocol to be the tensor product of the individual qubits prepared by the verifier, comprising the dotted-complete graph before the prover applies controlled-$Z$ to entangle them (these include the computation, trap and dummy qubits).

The fact that FK is robust means that the protocol's input state can be deviated from its ideal value by some small amount and the protocol will continue to function.
In particular, this input state could be the output of some other protocol, provided that this state was close to its ideal value. As we will see in the next subsection, the RUV protocol is capable of such a preparation. We start by formally defining robustness in this context.

\begin{definition}[Robustness] \label{def:robust}
A verification protocol with quantum input is robust if, given that the protocol input is $\epsilon$-close in trace distance to the ideal input,
in the limit where $\epsilon \rightarrow 0$ the completeness and soundness bounds remain unchanged.
\end{definition}

Mathematically, if we denote the multi-qubit input state as $\rho$, and the pure states comprising the ideal input as $\pi_i$, where $i$ goes from $1$ to the number of qubits, we have that:
\begin{equation}
\| \rho - \bigotimes_{i} \pi_i \|_{Tr} \leq \epsilon
\end{equation}
Note that $\rho$ is of the same dimension as $\bigotimes_{i} \pi_i$ as it does not contain any ancilla qubits from the environment.
Given the definition of robustness, we prove that:

\begin{theorem} \label{t:robust} \label{t-robust}
The FK protocol is robust and given an input which is $\epsilon$-close to its ideal value, the completeness is lower bounded by  $1-2\epsilon$ and the soundness bound changes by at most $O(\sqrt{\epsilon})$.
\end{theorem}

\noindent Because we are tracing out the environment, which could be controlled by an adversary, the security of the protocol, with a deviated input state, needs to be re-established. We highlight this in the following proof sketch of Theorem~\ref{t-robust}:

\begin{proof}[Proof sketch]
We first examine soundness which considers the case of a dishonest prover. Intuitively, when the prover is malevolent, he will try
to convince the verifier to accept an incorrect outcome and thus deviate from the correct protocol.
However, as shown in \cite{fk}, no matter how much the prover deviates, the probability for the verifier to accept a wrong outcome
is bounded.
If the input to the protocol
is already deviated from the ideal, one could expect that the soundness bound remains unchanged. The effect of a deviated input could be incorporated in the deviated actions of the prover.
This is indeed the case when the input is uncorrelated with any external system and we can express the deviation as
a CPTP map (see Lemma \ref{l-weakrobust} for detailed proof).

In the general case, however, the deviated input could be correlated with subsystems controlled by adversaries. This deviation could be used by the prover to improve his cheating probability. Mathematically this is manifested by the fact that the prover's action in the presence of initial correlations is not in general a trace preserving map. It can be expressed as a linear combination of a CPTP deviation and an inhomogeneous term  which could be either positive or negative as shown in the \cite{correlations}. In this case, we use the $\epsilon$-closeness of the input state to derive a bound of order $O(\sqrt{\epsilon})$ for the norm of the inhomogeneous term. From linearity, and using the previous argument it follows that in the general case the soundness bound changes by at most
$O(\sqrt{\epsilon})$ (see Lemma \ref{l-strongrobust} for detailed proof).

In the case of completeness, we are assuming the prover is honest.
If we start with an $\epsilon$-close input state, because of the linearity of the operators involved, we will end up with an output
state that is $O(\epsilon)$-close to the ideal output (see Lemma \ref{l:completeness} for detailed proof).
\end{proof}

A similar approach to Lemma \ref{l-weakrobust} was used in \cite{vfk} for defining approximate blindness, and in \cite{DFPR13} to prove universal composability  for blind quantum computing protocols.
However, to our knowledge, these results are not strong enough to cover the requirements for the composition with the RUV protocol. In \cite{vfk} only the blindness property was examined while verifying the computation was not considered. In \cite{DFPR13} they considered local-verifiability which does not take into account for example, the possibility of correlated attacks such as those that are possible when the two provers have a prearranged correlated strategy.

\subsection{Composition}\label{sect:main-comp}
One of our main objectives is to construct a device independent version of the FK protocol. The first step, was to show that FK is robust. This property guarantees that if we have an input state that is only approximately the ideal one, the protocol continues to work. We can now break the task of achieving device independent FK into two parts, which we need to compose sequentially. 
\begin{enumerate}
\item \textbf{State Preparation} - use a device independent protocol to prepare on the prover's side a state which is $\epsilon$-close to the FK input.
\item \textbf{Verified Delegated Computation} - run the FK protocol with the prover that has the $\epsilon$-close input state (since robustness allows this).
\end{enumerate}
The advantage of this technique is that we are free to use any protocol for state preparation as long as we have the guarantee of $\epsilon$-closeness. 
This is due to our strong robustness result, which shows that FK will work even if the deviation in the prepared state is correlated with the prover's cheating strategy in the delegated computation stage.
In this paper, we achieve state preparation using the device-independent state tomography sub-protocol of RUV. 
This sub-protocol has the $\epsilon$-closeness property that we require, as explained in \cite{ruv}. The resulting composite protocol will have a better round complexity than the full RUV protocol for the verification of quantum computations. The complexity can be improved further if a more efficient state preparation protocol is used. Recently, in an independent work that simultaneously appeared with our arxiv version, a more efficient scheme for state preparation is proposed that is based on a self-testing approach \cite{joe} rather than the rigidity of CHSH games \cite{ruv}.

We first clarify some details of the RUV protocol, which are essential in understanding how our composite protocol will work. RUV uses the rigidity property of CHSH games to determine that the provers share multiple copies of the Bell state $\ket{\Phi^+} = (\ket{00} + \ket{11}) / \sqrt{2}$, which is $X\!Z$-determined. They can then use $X\!Z$ state tomography to verify the preparation of any other $X\!Z$-determined state. In particular, they use it to tomographically verify the preparation of a set of states which can be used to perform universal computation. They also describe how it is possible to extend the protocol in order to have full tomography with the $Y$ operator as well \cite{ruv}. However, because they are using the $\ket{\Phi^+}$ Bell state, it is only possible to fix the $Y$ operator up to a sign change. That is, the provers can always choose to measure in either the $Y$ or $-Y$ bases without being detected (this corresponds to complex conjugating the states with respect to their representations in the computational basis).
In fact this problem has been noticed by others as well \cite{mckague, MM11}.
As explained in \cite{ruv}, it is possible to force the provers to consistently choose either $Y$ or $-Y$ for their measurements. This makes the resulting state prepared by state tomography close to either the ideal state or the complex conjugate of the ideal state.

At first glance it would seem that this could be problematic for the FK protocol.
We would have to show that running the FK protocol with an input state that is close to the complex conjugated version of the ideal input would be detected by the verifier. Intuitively this is the case, since trap qubits are in the $X\!Y$-plane and complex conjugating them would lead to different measurement outcomes. We will not prove this and instead provide a simpler solution.

The problem stems from the fact that we are using the $X\!Z$-determined $\ket{\Phi^+}$ state. Let us instead consider the state $\ket{\Psi^+} = (\ket{01} + \ket{10})/\sqrt{2}$. Using Theorem~\ref{thm:stabilizer} from \cite{ruv},
and the fact that $\ket{\Psi^+}$ has stabilizer generator set $\{X\otimes X,Y\otimes Y\}$ which belongs to
$\{ I, X, Y \}^{\otimes 2}$ we have that this state is $XY$-determined.

\begin{theorem} \label{thm:stabilizer}
\cite{ruv}
A stabilizer state is determined by any of its sets of stabilizer generators.
\end{theorem}

\noindent In principle it is possible to run a form of the RUV protocol in which we choose the CHSH games such that we rigidly determine that the provers share multiple copies of the Bell state $\ket{\Psi^+}$ instead of $\ket{\Phi^+}$. Analogous to the previous case, the extended form of the protocol would then fix the $Z$ operator up to a sign change (instead of the $Y$ operator). This means that the provers can always perform a reflection about the $X\!Y$ plane with no noticeable changes. However, the $X\!Y$ plane states are invariant under such a reflection. We can therefore use this to prepare the input which will be used by the FK protocol. The only problem we encounter is that we also require the preparation of $\ket{0}$ and $\ket{1}$ states which act as dummy qubits in the FK protocol \cite{fk}.
As described in Protocol~\ref{prot:fk} these dummy qubits are measured in order to ``break'' the dotted-complete graph into the computation graph and the two trap graphs.
The problem is that the $X\!Y$ plane reflection has the effect of flipping the computational basis states (state $\ket{0}$ becomes $\ket{1}$ and state $\ket{1}$ becomes $\ket{0}$). However this deviation (flip) has to be applied globally otherwise it affects the statistics of the CHSH game and thus the verifier rejects \cite{ruv}.
Such a global flip is detected by the FK protocol.
A formal proof is given in Lemma \ref{l:reflectedFK}, Section~\ref{sect:composite}, while below we give a sketch of the proof.

In the honest scenario for the FK protocol, the measurement of a dummy qubit in state $\ket{1}$ introduces an additional $Z$ correction to its neighbouring qubits (this is because we are using the controlled-$Z$ operation for entangling qubits). Hence in a malicious setting the effect that a flip has on a trap qubit with an odd number of neighbouring dummy qubits, leads to an extra $Z$ operation. Such a $Z$ flip changes a $\ket{+_\theta}$ state to $\ket{-_\theta}$. Thus, the measurement of this trap qubit will deterministically fail and the verifier will detect this. On the other hand, since the verifier chooses the input, he can always pad the computation such that the overall graph has trap qubits with an odd number of neighbour dummy qubits. This is due to the fact that in a dotted-complete graphs (Definition \ref{d-dotted}), some of the traps will have $N - 1$ neighbouring dummy qubits. Therefore, if the size of the input computation $N$ is odd, the verifier need only pad the computation size to become $N + 1$.

Now we are in a position to construct the composite protocol which composes RUV with FK.
We give a modified version of the state tomography protocol of RUV (see Protocol~\ref{prot:modifiedstatetomography}). Proof that Protocol~\ref{prot:modifiedstatetomography} is valid verification protocol is given in Section~\ref{sect:composite}. The purpose of this modification is to verifiably prepare the minimal resource states which are subsequently used as inputs for the FK protocol.

\begin{algorithm}[h!]
\caption{Modified State Tomography Protocol}
 \label{prot:modifiedstatetomography}
 \vskip 0.2 cm
\textbf{Assumptions} \\
Let $S = \{ (1,0,0), (0,1,0), (0,0,1), \frac{1}{\sqrt 2} (1, 1, 0), \frac{1}{\sqrt 2} (1, -1, 0), \frac{1}{\sqrt 2} (1, 0, 1), \frac{1}{\sqrt 2} (1, 0, -1), \frac{1}{\sqrt 2} (0, 1, 1),$ $\frac{1}{\sqrt 2} (0, 1, -1) \}$.
Let $M_v$ be a $2$ outcome projective measurement defined by the projectors: $\frac{1}{2} (I + \vec v \cdot (X, Y, Z))$
and $\frac{1}{2} (I - \vec v \cdot (X, Y, Z))$. \\
Let the tuple $(\vec{a}, \vec{b}) \in S \times S$ denote the measurements $M_a$ for prover $1$ and $M_b$ for prover $2$ that they need to perform on their halves on an entangled state when instructed by the verifier.
Sets of such tuples define CHSH games. For example the set $\{ (1,0,0), (0,0,1) \} \times \{ \frac{1}{\sqrt 2} (1, 0, 1), \frac{1}{\sqrt 2} (1, 0, -1) \}$ defines the $X\!Z$ CHSH game.
Given $S$, there are six such sets of CHSH games (two $X\!Z$, two $X\!Y$ and two $Y\!Z$) \cite{ruv}.
For a suitable numbering of these games, let $CHSH_i$ be the $i$th CHSH game, $i \in \{ 1, ... 6 \}$.
\\

\textbf{Protocol} \\
The verifier alternates uniformly at random between the following subprotocols:
\begin{enumerate}
\item \textbf{CHSH games.} Verifier referees $6N$ sets of sequential CHSH games, such that each group of $N$ sets is one of the six possible
CHSH types of games. Each set consists of $n_g$ games between prover $1$ and prover $2$.
For each group of $N$ CHSH games the verifier rejects if the two provers win less than:
\[
\cos^2(\pi/8) Nn_g - \frac{1}{2\sqrt{2}}\sqrt{N n_g \log(N n_g)}
\]
of the games.
\item \textbf{State tomography.} Verifier
chooses $K \in [N]$ uniformly at random and also randomly chooses $CHSH_i$ as one of the six possible CHSH games.
Then he referees $K - 1$ sets of $CHSH_i$ games, sending
the questions from the $K$th set to prover $1$, while running a state tomography protocol with prover $2$.
In this protocol prover $2$ is asked to prepare 
resource states by measuring his halves of the shared Bell states.
This will collapse prover $1$'s states to the same 
resource states up to corrections. In the context of composition, these resource states will constitute the FK input. The verifier uses the measurement outcomes of prover $1$ to
tomographically check this preparation. He rejects if the tomography statistics are inconsistent. In the end, if the verifier accepts, he concludes that with high probability prover $1$ has a state which is close in trace distance to the tensor product of resource states. The formal statement of this fact is given in Theorem~\ref{t:statetomography}, taken from \cite{ruv},
and more precisely in Equation~\ref{eqn:closeness} from Lemma~\ref{l:closeness}.
\end{enumerate}
\end{algorithm}

The composite protocol, given as Protocol~\ref{prot:composite}, is the sequential composition of the modified state tomography of Protocol \ref{prot:modifiedstatetomography} with both provers followed by the FK protocol with prover $1$. Note that since prover $1$ is involved in both state tomography as well as the FK protocol, the strong version of the robustness property is required. This is to address the effect of any potential correlated attacks where provers $1$ and $2$ have agreed in advance on a strategy.
The deviations of prover $2$, in the preparation stage, could be correlated with the deviations of prover $1$ during the computation stage (FK). This is the first rigorous proof of a protocol that involves lifting the FK protocol to the entangled provers setting. We give here the correctness and soundness of this protocol and show that it is more efficient than the RUV protocol (Theorem~\ref{t:composite}) while in Section~\ref{sect:composite} we give the proof of this theorem.

\begin{algorithm} [t!]
\caption{Composite Verification Protocol}
\label{prot:composite}
\vskip 0.2 cm
\begin{enumerate}
\item Run the modified state tomography protocol (Protocol ~\ref{prot:modifiedstatetomography}).
\item From the states prepared by this protocol on prover $1$'s side,
select the input for FK and run the FK protocol with prover $1$. (Protocol ~\ref{prot:fk})
\end{enumerate}
\end{algorithm}

\begin{theorem} \label{t:composite}
Assuming the verifier wants to delegate the computation of a quantum circuit of size $n$, Protocol~\ref{prot:composite} is a \MIPs verification protocol having completeness lower bounded by $1 - O(n^{-1/128})$,
soundness upper bounded by
$\left( \frac{2}{3} \right)^{\lceil \frac{2d}{5} \rceil} +
O(n^{-1/12})$, where $d$ is the security parameter of the FK protocol, and round complexity $O(n^c)$,
where there exists some constant $c$ such that $c > 2048$.
\end{theorem}

While the obtained round complexity is an improvement over RUV (Lemma~\ref{l:ruv-complexity}) it is still far from practical. However, we believe our approach serves as a proof of principle, that this type of composition can be beneficial. It also highlights where improvements could be made. It is the state tomography subprotocol that increases the round complexity, while the FK protocol has a relatively low complexity\footnote{Note that the round complexity of FK could be further reduced to linear, if one is willing to admit a higher upper bound for soundness.}. The detailed proofs are given in Section~\ref{sect:composite}.

\subsection{Fault Tolerance} \label{sect:main-ft}

In constructing our composite verification protocol, we used the robustness of the FK protocol. Our last result is to characterise the difference between robustness and fault tolerance and to show that the FK protocol can be made fault tolerant using a topological error correcting code. Consequently, our composite protocol can also be made fault tolerant provided that the state tomography part is run on top of an error correcting code.

As mentioned before robustness is a protocol's ability to continue to function given a deviated input. Fault tolerance is when a protocol functions correctly in the presence of error prone devices. The essential assumption for robustness is that the actual (multi-qubit) input is $\epsilon$-close to its ideal value. Fault tolerant protocols, on the other hand, assume that errors can occur at
each individual qubit. The faulty devices are usually represented by the action of a partially depolarizing channel: $\mathcal{E} = (1 - p)[I] + \frac{p}{3} ([X] + [Y] + [Z])$.
Here $p$ is the probability of error, and the square brackets indicate the action of an operator.
This leads to the following observation:

\begin{lemma} \label{l:ft1}
Let $\sigma = \otimes_{i=1}^n \rho_i$ be a system of $n$ qubits. Assume each qubit goes through a partially depolarizing channel $\mathcal{E}$ having probability of error $p > 0$.
Let the state of the system, after all qubits have passed through the channel, be $\sigma' = \otimes_{i=1}^n \mathcal{E}(\rho_i)$. We have that $\|\sigma - \sigma'\|_{Tr} \leq min(1, np)$ and there exist states $\sigma$ for which $\|\sigma - \sigma'\|_{Tr} = 1$.
\end{lemma}

This means that the deviation of an $n$-qubit system from the ideal input is not bounded by some constant amount.
This is intuitively clear, since by adding more qubits, we introduce more errors and the state of the composite system is further from its intended value.
In contrast to this, when considering robustness, the distance between the actual and ideal state is bounded by an arbitrarily small quantity. We will now address how can we do verification
in an error prone setting.

\begin{lemma} \label{l:ft2}
Assume we run the FK protocol with $N_T$ traps and each qubit is subject to the action of a partially depolarizing channel $\mathcal{E}$ having
probability of error $p > 0$.
Given the simplifying assumption that if a qubit is changed (through the action of an $X$, $Y$ or $Z$ operator) it will produce an incorrect measurement
outcome, the completeness of this protocol is upper bounded by $(1 - p)^{N_T}$.
\end{lemma}

\begin{algorithm}[t!]
\caption{Fault Tolerant FK Protocol}
\label{prot:ft}
\vskip 0.2 cm
\textbf{Assumptions} \\
The verifier wants to compute the execution of a measurement graph $G$ having $n$ qubits.
Both the verifier and prover's devices are subject to noise modelled as a partially depolarizing channel
acting on the preparation of the qubits and the application of the quantum gates.
For single qubits the channel is described by:
\begin{equation}
\mathcal{E}_1 = (1 - p)[I] + \frac{p}{3} ([X] + [Y] + [Z])
\end{equation}
And for two qubit states by:
\begin{equation}
\mathcal{E}_2 = (1 - p)[I \otimes I] + \frac{p}{15} ([I \otimes X] + \cdots + [Z \otimes Z])
\end{equation}
Additionally assume $p \leq p_{correct}$, where $p_{correct}$ is a threshold such that depolarizing noise bellow this threshold is corrected by the
topologically protected code from \cite{rhg}. \\
Let $\mathcal{G}^{\nu}$ denote a \emph{brickwork state} encoding the graph $G$ and containing one trap qubit, as explained in \cite{fk}.
Let $\mathcal{L}^{\nu}$ denote a fault tolerant encoding of the graph $\mathcal{G}^{\nu}$ using the topologically protected code from
\cite{rhg}. The encoding is done as explained in \cite{topo}, hence $\mathcal{L}^{\nu}$ will be decorated lattices (see Figures~1,~2 in \cite{topo}).
The index $\nu$ denotes the randomness in the $\theta$ angles for the encoding as chosen by the verifier.
Let $\mathcal{S^{\widetilde{\nu}}} = \mathcal{L}^{\nu_1} \otimes \mathcal{L}^{\nu_2} \otimes \cdots \otimes \mathcal{L}^{\nu_N}$, where
$R/\log \left( \frac{cn}{cn - 1} \right) < N < R/\log \left( \frac{cn}{cn - 1} \right) + O(1)$,
for some constants $R > 1$, $c > 2$ and $\widetilde{\nu} = \{ \nu_1 \cdots \nu_N \}$.
We will refer to $\mathcal{S^{\widetilde{\nu}}}$ as a \emph{sequence} of encodings. \\

\textbf{Protocol}
\begin{enumerate}
\item The verifier chooses $R > 1$ and constructs the random set $\widetilde{\nu}$.
\item The verifier
prepares the qubits for the sequence $\mathcal{S}^{\widetilde{\nu}}$ and sends them to the prover along with instructions on how to construct
$\mathcal{S}^{\widetilde{\nu}}$.
\item The verifier sends measurement instructions to the prover in order to compute the executions of the encoded graphs.
\item The prover sends the measurement outcomes to the verifier.
\item Steps 3 and 4 repeat until the verifier either accepts or rejects.
\end{enumerate}
The verifier rejects if any of the traps fail.
He takes the outcome of the computation to be the majority outcome over all computations (graphs $\mathcal{L}^{\nu_i}$).
\end{algorithm}

\noindent It is evident that assuming faulty devices where each qubit behaves as if it crossed a partially depolarizing channel, the completeness of the protocol becomes exponentially small (as function of the number of traps). This is clearly unsatisfactory. The arguably simplest solution would be to alter the acceptance condition of the protocol. Since it is unlikely that all trap measurements
succeed, even for honest prover, the verifier should accept a result if the traps that succeed are above some fixed fraction.

\begin{lemma} \label{l:ft3}
Assume we run a modified FK protocol with $N_T$ traps and each qubit is subject to the action of a partially depolarizing channel $\mathcal{E}$ having
probability of error $p > 0$.
The modification is that the verifier accepts if there are fewer than $N_T(p + \epsilon)$ mistakes at trap measurements, where $\epsilon > 0$ is a suitably chosen small number. The completeness of this protocol is lower bounded by $1 - \exp(-2\epsilon^2 N_T)$.
\end{lemma}

\noindent The above modification resolves the issue raised regarding the completeness bound. However, if we were to make such modification, we have the following consequence for the soundness of the protocol:

\begin{lemma} \label{l:ft4}
Assume we run a modified FK protocol with $N_T$ traps, $N$ qubits in total
and each qubit is subject to the action of a partially depolarizing channel $\mathcal{E}$ having
probability of error $p > 0$.
The modification is that the verifier accepts if there are fewer than $N_T(p + \epsilon)$ mistakes at trap measurements, where $\epsilon > 0$ is a suitably chosen small number. The soundness of this protocol is upper bounded by ${N_T \choose N} \left( \frac{2}{3} \right) ^{\lceil \frac{2d}{5} \rceil}$.
\end{lemma}

\noindent We can see that introducing a threshold of acceptance leads to an increased bound on soundness. Again, expected, since we allow the prover to tamper with some of the traps (and, by extension, with the computation as well) without rejecting the output.
To solve these problems we need to use a fault tolerant code.
The FK protocol already uses a fault tolerant code to encode the computation graph. However this is done in order to boost the value of the soundness parameter. The trap qubits are not encoded with the code (only the computation is).
Thus we propose a modified FK protocol. This is described in Protocol~\ref{prot:ft}.
In this protocol we encode both computations and traps in a fault tolerant code and use sequential repetitions (also used in \cite{efk}). This leads to our final main result:

\begin{theorem} \label{t:ft5}
Under the assumption of a faulty setting where qubit preparation and quantum gates are subject to partially depolarizing noise having bounded
probability $p$, Protocol~\ref{prot:ft} is a valid verification protocol having completeness $1$, soundness upper bounded by
$(1/2)^R$, where $R$ is a constant such that $R > 1$. The protocol has round complexity $O(n^2)$.
\end{theorem}

\noindent The proof of this theorem (and previous lemmas) are given in Section~\ref{sect:ft}.
An important point to make is that the composite protocol we constructed can also be made fault tolerant. To achieve this the state tomography protocol should be run on top of a fault tolerant code.
As mentioned in \cite{ruv}, in principle, this is straightforward for blind, verified computation, since the
provers can work on top of a quantum error-correcting code and entanglement can be distilled with the help of the verifier \cite{Sheng2015, distillation}.

\section{Proof of Robustness} \label{sect:robust}
In this section we prove the robustness of the FK protocol.
We start by first proving a simpler result, namely the robustness of the protocol under the assumption that the input is uncorrelated with any external system. We then remove this assumption and use our results to prove the main theorem, necessary for the composition with the RUV protocol.

\begin{lemma} \label{l:soundness}\label{l-weakrobust}
If the initial input state of the FK protocol is $\epsilon$-close to the ideal input state and uncorrelated with any external system,
the soundness bound does not change.
\end{lemma}
\begin{proof}
We will follow the same proof technique as in \cite{fk} and show that the soundness bound does not change. This is done by incorporating
the assumption of a deviated input into that proof.
The outcome density operator of the protocol is denoted $B_j(\nu)$, where $\nu$ denotes the verifier's choices of input
variables and $j$ ranges over the prover's choices of possible actions ($j=0$ is the correct/honest action).
If the outcome is incorrect it means that all of the traps have passed, but the computation is not correct. This is associated
with the following projection operator \cite{fk}:
\begin{equation}
P_{incorrect} = (\mathbb{I} - \Ket{\Psi_{ideal}}\Bra{\Psi_{ideal}}) \bigotimes_{t\in T} \ket {\eta_t^{\nu_T}} \bra {\eta_t^{\nu_T}}
\end{equation}
Here, $\Ket{\Psi_{ideal}}\Bra{\Psi_{ideal}}$ is the ideal output state, and
$\bigotimes_{t\in T} \ket {\eta_t^{\nu_T}} \bra {\eta_t^{\nu_T}}$
is the state associated with the trap qubits. Notice that we are projecting to a state in which the output is orthogonal to its ideal value,
and the traps are correct. This expresses the fact that the verifier will accept an incorrect computation.
The associated probability for that event is $p_{incorrect}$ and can be expressed as:
\begin{equation}
p_{incorrect} = \sum\limits_{\nu} p(\nu) Tr(P^{\nu}_{incorrect} B_j(\nu))
\end{equation}
Which is a weighted average of the incorrect outcome probabilities (expressed by the trace operator) over all possible input states.
The outcome density operator can be written as:
\begin{equation} \label{eqn:fkoutcome}
B_j(\nu) = Tr_P \left( \sum\limits_{b} \Ket{b+c_r}\Bra{b} C_{\nu_C,b} \Omega P
\overbrace{(\underbrace{(\otimes^P \Ket{0}\Bra{0})}_\text{Prover's qubits} \otimes
\underbrace{\Ket{\Psi^{\nu,b}} \Bra{\Psi^{\nu,b}}}_\text{Input state})}^\text{Joint system state $\sigma^{\nu,b}$}
P^{\dagger} \Omega^{\dagger} C_{\nu_C, b}^{\dagger} \Ket{b} \Bra{b + c_r} \right)
\end{equation}
Notice the following, as explained in \cite{fk}:
\begin{itemize}
\item We are tracing over the prover's qubits;
\item We have denoted the joint state, comprised of the input and the prover's qubits, as $\sigma^{\nu,b}$;
\item $j$ ranges over the prover's possible strategies ($j = 0$ is the honest strategy);
\item $b$ indicates the possible branches of computation parametrised by the measurement results sent by the prover to the verifier;
\item $c_r$ indicates corrections that need to be performed on the final, classical output due to the MBQC computation together with the random phase introduced by the verifier;
\item $P$ is the computation that we want the prover to do;
\item $\Omega$ is the prover's deviation from the desired computation;
\item $C_{\nu_C, b}$ are the corrections the prover applies to its quantum output depending on the measurement outcomes (as in the measurement-based model);
\end{itemize}
We now need to incorporate the approximate input state into this operator. We will not use the $\epsilon$-closeness of the deviated state to the ideal one, and prove a stronger result, that the soundness bound does not change \emph{regardless} of the input state. Concretely,
assume the deviated input is:
\begin{equation}
\rho^{\nu,b} = \mathcal{E} \left( \Ket{\Psi^{\nu,b}} \Bra{\Psi^{\nu,b}} \right)
\end{equation}
Where $\mathcal{E}$ is a CPTP map which represents any deviation from the ideal input state either from incorrect
preparation, a malicious prover or faulty devices. This is equivalent to applying some unitary $U$ to the input state tensored
with some environment qubits that are traced out. We can express this mathematically as:
\begin{equation} \label{eqn:rho}
\rho^{\nu,b} = Tr_E( U \left( (\otimes^{E} \Ket{0} \Bra{0}) \otimes \Ket{\Psi^{\nu,b}} \Bra{\Psi^{\nu,b}} \right) U^{\dagger})
\end{equation}
The joint system state $\sigma^{\nu,b}$ becomes\footnote{Here and in the following expressions we've used the fact that the
partial trace is linear and can therefore be moved outside.}:
\begin{equation}
\sigma^{\nu,b} = Tr_E( (\otimes^P \Ket{0} \Bra{0}) \otimes
U \left( (\otimes^{E} \Ket{0} \Bra{0}) \otimes \Ket{\Psi^{\nu,b}} \Bra{\Psi^{\nu,b}} \right) U^{\dagger})
\end{equation}
Let us consider a new unitary $V = (\mathbb{I} \otimes U)$. This allows us to rewrite the joint system state as:
\begin{equation}
\sigma^{\nu,b} = Tr_E (V \left( (\otimes^{E+P} \Ket{0} \Bra{0}) \otimes \Ket{\Psi^{\nu,b}} \Bra{\Psi^{\nu,b}} \right) V^{\dagger})
\end{equation}
Since $P$, the computation, is a unitary operator, there must exists some unitary $V'$ such that $V = P^{\dagger} V' P$.
Substituting this into the previous expression gives us:
\begin{equation}\label{e-deviate}
\sigma^{\nu,b} = Tr_E (P^{\dagger} V' P
\left( (\otimes^{E+P} \Ket{0} \Bra{0}) \otimes \Ket{\Psi^{\nu,b}} \Bra{\Psi^{\nu,b}} \right)
P^{\dagger} V'^{\dagger} P)
\end{equation}
Incorporating Equation~\ref{e-deviate} into the expression for $B_j(\nu)$ in Equation~\ref{eqn:fkoutcome} we obtain:

\begin{equation}
\begin{split}
B_j(\nu) = Tr_P \Bigl( \sum\limits_{b} & \Ket{b+c_r}\Bra{b} C_{\nu_C,b} \Omega P \\
& (Tr_E (P^{\dagger} V' P \left( (\otimes^{E+P} \Ket{0} \Bra{0}) \otimes \Ket{\Psi^{\nu,b}} \Bra{\Psi^{\nu,b}} \right) P^{\dagger} V'^{\dagger} P)) \\
& P^{\dagger} \Omega^{\dagger} C_{\nu_C, b}^{\dagger} \Ket{b} \Bra{b + c_r} \Bigr)
\end{split}
\end{equation}

The assumption of the lemma is that the input state is not correlated with any external system. Hence, the spaces $E$ and $P$ are independent. This means that the prover's deviation, $\Omega$, is not acting on $E$ and therefore we can ``push'' the inner trace operator to the beginning of the equation. Also using the fact that $PP^{\dagger} = P^{\dagger}P = \mathbb{I}$,
we obtain:

\begin{equation}
\begin{split}
B_j(\nu) = Tr_{P+E} \Bigl( \sum\limits_{b} & \Ket{b+c_r}\Bra{b} \\
& C_{\nu_C,b} \Omega V' P ((\otimes^{P+E} \Ket{0}\Bra{0}) \otimes \Ket{\Psi^{\nu,b}} \Bra{\Psi^{\nu,b}}) P^{\dagger} V'^{\dagger} \Omega^{\dagger} C_{\nu_C, b}^{\dagger} \\
& \Ket{b} \Bra{b + c_r} \Bigr)
\end{split}
\end{equation}

We can now include the input deviation given by $V'$ into the prover deviation $\Omega$, by considering $\Omega' = \Omega V'$. This is possible because we are bounding the probability over all possible
deviations, $\Omega$, of the prover in the computation and all possible deviations, $V'$, from the preparation part. Thus, we can consider this to be a single, global, deviation given by $\Omega'$.

\begin{equation}
\begin{split}
B_j(\nu) = Tr_{P+E} \Bigl( \sum\limits_{b} & \Ket{b+c_r}\Bra{b} \\
& C_{\nu_C,b} \Omega' P
((\otimes^{P+E} \Ket{0}\Bra{0}) \otimes
\Ket{\Psi^{\nu,b}} \Bra{\Psi^{\nu,b}})
P^{\dagger} \Omega'^{\dagger} C_{\nu_C, b}^{\dagger} \\
& \Ket{b} \Bra{b + c_r} \Bigr)
\end{split}
\end{equation}

As a result, the above equation has the same form as the undeviated input scenario of Equation~\ref{eqn:fkoutcome}.
This makes sense since all we have done is to incorporate the deviation
of the input into the prover's cheating strategy. The original proof continues as it is in \cite{fk},
and the bound remains unchanged
\begin{equation}
p_{incorrect} \leq \left( \frac{2}{3} \right)^{\lceil \frac{2d}{5} \rceil}
\end{equation}
\end{proof}

The type of robustness guaranteed by this lemma is not sufficient to prove the security of any protocol that composes RUV with FK. For example if we use prover $2$ of RUV to prepare the input of the FK protocol for prover $1$, this input is in general correlated with prover $2$'s system.
To address this issue, we use from \cite{ruv} the following corollary of the gentle measurement lemma and the special Kraus representation in the presence of initial correlations given in \cite{correlations}.

\begin{corollary} \label{c:gentlemeasurement}
\cite{ruv} Let $\rho$ be a state on $\H_1 \otimes \H_2$, and let~$\pi$ be a pure state on~$\H_1$.
If for some $\delta \geq 0$, $\Tr (\pi \Tr_2 \rho) \geq 1 - \delta$, then
\begin{equation}
|| \rho - \pi \otimes \Tr_1 \rho ||_{Tr} \leq 2 \sqrt \delta + \delta
\end{equation}
\end{corollary}

\begin{lemma} \label{l:soundness_corr}\label{l-strongrobust}
If the initial input state of the FK protocol is $\epsilon$-close to the ideal input state
the soundness bound changes by at most $O(\sqrt{\epsilon})$.
\end{lemma}

\begin{proof}
Consider a composite correlated state $\rho_{AB}$ where systems $A$ and $B$ are not communicating and let $\rho_A = Tr_B(\rho_{AB})$ and $\rho_B = Tr_A(\rho_{AB})$. If $\rho_A$ is used as input for the FK protocol,  the existence of correlations (not present in the previous lemma) can be exploited by an adversarial prover. Hence the deviation can no longer be expressed as a CPTP map over this subsystem.  As it is shown in \cite{correlations}, in presence of initial correlations defined as:
\begin{equation}
\rho_{corr} = \rho_{AB} - \rho_A \otimes \rho_B
\end{equation}
the evolution of the subsystem $\rho_A$ is the following:
\begin{equation} \label{eqn:noncptp}
\rho_A \rightarrow \mathcal{E}(\rho_A) + \delta\rho_A
\end{equation}
Here $\mathcal{E}$ is a CPTP map and $\delta\rho_A$ is an inhomogeneous term which is added to the CPTP evolution due to the presence of initial correlations. In addition we have the following property:
\begin{equation}
\delta\rho_A = Tr_B(U_{AB} \rho_{corr} U^{\dagger}_{AB})
\end{equation}
We can see that substituting $\rho_A$ in the outcome density operator of the FK protocol gives different soundness bound than the one in Lemma~\ref{l:soundness}. The difference stems from the extra $\delta\rho$ term. However we can use the fact that $\rho$ is $\epsilon$-close to the ideal state (lemma assumption) to show that the norm of $\delta\rho_A$ is at most of order
$O(\sqrt{\epsilon})$. To prove this we first find a bound for the norm of $\rho_{corr}$, and since $\delta\rho_A$ is just a CPTP map applied to $\rho_{corr}$
it follows that the norm of $\delta\rho_A$ has the same bound.
Moreover, the action of the FK protocol can be modelled as a CPTP map, therefore acting on $\delta\rho_A$ will not increase the norm.
It follows that the overall soundness bound changes by at most $O(\sqrt{\epsilon})$.

If we denote the ideal state as $\Ket{\psi}$, we know that:
\begin{equation}
|| \rho_A - \Ket{\psi}\Bra{\psi} ||_{Tr} \leq \epsilon
\end{equation}
It is also known, from the relationship between fidelity and trace distance, that:
\begin{equation}
1 - \Bra{\psi} \rho_A \Ket{\psi} \leq || \rho_A - \Ket{\psi}\Bra{\psi} ||_{Tr}
\end{equation}
Combining these two yields:
\begin{equation}\label{e-trace}
\Bra{\psi} \rho_A \Ket{\psi} \geq 1 - \epsilon
\end{equation}
Recall that $Tr(\Ket{\psi}\Bra{\psi} \rho_A) = \Bra{\psi} \rho_A \Ket{\psi}$, using Equation~\ref{e-trace} and Corollary~\ref{c:gentlemeasurement} (where $\rho$ is substituted with $\rho_{AB}$ and $\pi$  with $\ket {\psi} \bra{\psi}$) we have:
\begin{equation}
|| \rho_{AB} - \Ket{\psi}\Bra{\psi} \otimes \rho_B ||_{Tr} \leq 2 \sqrt \epsilon + \epsilon
\end{equation}
The trace norm of $\rho_{corr}$ is simply the trace distance between $\rho_{AB}$ and $\rho_A \otimes \rho_B$ as can be seen from the definition.
Using the triangle inequality, we have:
\begin{equation}
|| \rho_{AB} - \rho_A \otimes \rho_B ||_{Tr} \leq || \rho_{AB} - \Ket{\psi}\Bra{\psi} \otimes \rho_B ||_{Tr} +
|| \Ket{\psi}\Bra{\psi} \otimes \rho_B - \rho_A \otimes \rho_B ||_{Tr}
\end{equation}
For the last term, using the additivity of trace distance with respect to tensor product, we get:
\begin{equation}
|| \Ket{\psi}\Bra{\psi} \otimes \rho_B - \rho_A \otimes \rho_B ||_{Tr} \leq
|| \Ket{\psi}\Bra{\psi} - \rho_A ||_{Tr} + || \rho_B - \rho_B ||_{Tr} = \epsilon
\end{equation}
Combining these last three inequalities we obtain:
\begin{equation}
|| \rho_{AB} - \rho_A \otimes \rho_B ||_{Tr} \leq 2\sqrt{\epsilon} + 2\epsilon
\end{equation}
Since $0 \leq \epsilon \leq 1$, the bound is of order $O(\sqrt{\epsilon})$. We have therefore bounded the norm of $\rho_{corr}$ and thus the norm of $\delta \rho_A$.

We can now take our expression for the deviated input
from Equation~\ref{eqn:noncptp} and substitute it into Equation~\ref{eqn:rho},
from Lemma~\ref{l:soundness}. Since trace is a linear operation, it will result in the addition of an inhomogeneous term to each
equation that involves the outcome density operator.
But since the inhomogeneous term has bounded trace norm, and the action of the outcome density operator
is trace preserving, it follows that we obtain the same bound as in Lemma~\ref{l:soundness} with the addition of an extra term of order $O(\sqrt{\epsilon})$.
This concludes the proof.
\end{proof}

\begin{lemma} \label{l:completeness}
If the initial input state of the FK protocol is $\epsilon$-close to the ideal input state,
the completeness is lower bounded by $1 - 2\epsilon$.
\end{lemma}
\begin{proof}
In the simplest sense, the FK protocol can be abstractly thought of as a CPTP map $\mathcal{P}$,
that takes some input state to an output state. Since we are assuming the prover is honest, the output state will be
$B_0(\nu)$. However, this is in the case where the input is assumed to be ideal. We are dealing with a deviated input,
hence our output state will be $B'_0(\nu)$. Writing these out explicitly we have:
\begin{equation}
B_0(\nu) = \mathcal{P}(\Ket{\Psi^{\nu}} \Bra{\Psi^{\nu}})
\end{equation}
\begin{equation}
B'_0(\nu) = \mathcal{P}(\rho^{\nu})
\end{equation}
Where, $\rho^{\nu}$ is the deviated input, and by assumption
\begin{equation}
\| \rho^{\nu} -  \Ket{\Psi^{\nu}} \Bra{\Psi^{\nu}}\|_{Tr} \leq \epsilon
\end{equation}
Note that in the following we do not need to consider non CPTP map evolution since the provers are assumed to behave honestly. Hence even in the presence of initial correlation, the subsystem will evolve according to the desired CPTP map of the protocol.
However CPTP maps cannot increase the trace distance, which leads to:
\begin{equation}
\| B_0(\nu) - B'_0(\nu)  \|_{Tr} \leq \| \Ket{\Psi^{\nu}} \Bra{\Psi^{\nu}} - \rho^{\nu} \|_{Tr} \leq \epsilon
\end{equation}
This also applies for projection operators and if in particular we consider $P_{correct}$, the projection onto the correct
output state, we also have that:
\begin{equation}
\| P_{correct} B_0(\nu) - P_{correct} B'_0(\nu)  \|_{Tr} \leq  \| B_0(\nu) - B'_0(\nu)  \|_{Tr} \leq \epsilon
\end{equation}
Next we use the reverse triangle inequality, which gives us:
\begin{equation}
\left| \; \| P_{correct} B_0(\nu) \|_{Tr} - \| P_{correct} B'_0(\nu)  \|_{Tr} \; \right| \leq
\| P_{correct} B_0(\nu) - P_{correct} B'_0(\nu)  \|_{Tr} \leq \epsilon
\end{equation}
And since we are dealing with positive definite operators, we know that:
\begin{equation}
\| P_{correct} B_0(\nu) \|_{Tr} = \frac{1}{2} Tr(P_{correct} B_0(\nu))
\end{equation}
\begin{equation}
\| P_{correct} B'_0(\nu)  \|_{Tr} = \frac{1}{2} Tr(P_{correct} B'_0(\nu))
\end{equation}
But $Tr(P_{correct} B_0(\nu)) = 1$ (the completeness when we have ideal input), so:
\begin{equation}
\left| 1 - Tr(P_{correct} B'_0(\nu)) \right| \leq 2 \epsilon
\end{equation}
Lastly, because $Tr(P_{correct} B'_0(\nu)) \leq 1$, we get:
\begin{equation} \label{eqn:ineq}
1 - 2 \epsilon \leq Tr(P_{correct} B'_0(\nu))
\end{equation}
Thus, the probability of accepting a correct outcome, under the assumption that the input state is $\epsilon$-close to the ideal
input, is greater than $1 - 2\epsilon$.
\end{proof}

It is now easy to see that the Proof of Theorem~\ref{t:robust} follows directly from Definition~\ref{def:robust} and Lemmas~\ref{l:soundness_corr} and~\ref{l:completeness}. Having the robustness property, the FK protocol can receive an input, which is $\epsilon$-close to its ideal value, from another protocol.
As we have shown, even if this input is correlated with an external system, we can still perform the verification as long as we have
$\epsilon$-closeness.

\section{Proof of Compositionality}\label{sect:composite}

To prove the security of the composite protocol (Theorem~\ref{t:composite}), we first need to prove that the FK protocol rejects with high probability a state close to a reflection about the $XY$-plane (Lemma~\ref{l:reflectedFK}). Then we prove that the modified state tomography protocol (Protocol~\ref{prot:modifiedstatetomography}), satisfies the $\epsilon$-closeness property required by the (robust) FK. This is achieved by showing Lemmas~\ref{l:modifiedstatetomography} and~\ref{l:closeness}. Finally we give the proof of Theorem~\ref{t:composite}.

\begin{lemma} \label{l:reflectedFK}
If the initial input state of the FK protocol is $\epsilon$-close to a reflection about the $X\!Y$-plane of the ideal input state,
the protocol will reject it with high probability.
\end{lemma}
\begin{proof}
First we note that the input to the FK protocol consists of $X\!Y$-plane states and dummy qubits which are either $\ket{0}$ or $\ket{1}$.
The $X\!Y$-plane states are invariant under the reflection, while the dummy states will be flipped. Assume that there is a trap that has an odd number of (dummy) neighbours. The verifier knows that he sent the state $\ket{+_\theta}$ and expects to make a $Z$ correction if the number of $\ket{1}$ neighbours is odd. However, if instead of what the verifier expects, there is an overall reflection with respect to the $X\!Y$-plane, then for each of the neighbours of the trap there will be a new $Z$ correction (the $\ket{0}$ will become $\ket{1}$ inducing a $Z$, while the $\ket{1}$ will become $\ket{0}$ undoing the previous $Z$ correction, which is equivalent with another $Z$ correction since $Z^2=\mathbf{1}$). Therefore, if the neighbours of a trap are odd in number, he will expect the exact opposite result and will deterministically detect the deviation. For this to happen it suffices that the verifier makes sure that at least one trap has odd number of neighbours, something that can be easily achieved.
Therefore, the FK protocol will always reject the reflected ideal input state.
Given that the input is $\epsilon$-close to this, we have shown in Lemma~\ref{l:soundness_corr} that the outcome density operator changes by at most $O(\sqrt{\epsilon})$ from its ideal value. Thus, the output state is $O(\sqrt{\epsilon})$ close to the reflected ideal input state. It follows that the protocol will reject this state with at most probability $1 - O(\sqrt{\epsilon})$.
\end{proof}

In proving the correctness of our protocol we first need to show the correctness of the modified state tomography protocol. Here we focus on the main results that we use for showing the correctness and security of this protocol.
We start with a theorem from \cite{ruv}:

\begin{theorem} \label{t:statetomography}
\cite{ruv}
Fix $\xzdeterminedset = \{ \pi^1, \ldots, \pi^{2^q} \}$ a complete, orthonormal set of $q$-qubit $X\!Z$-determined pure states.  For a sufficiently large constant~$\alpha$ and for sufficiently large~$n$, let $m = m(n) \geq q n$ and $N \geq m^{\alpha - 1}$.  Let $\sigma \in [m]^{q n}$ be a list of distinct indices.  Consider a combination of the following two protocols between the verifier, Eve, and the provers, Alice and Bob:
\begin{enumerate}
\item
CHSH games: In the first protocol, Eve referees $N m$ sequential CHSH games.  She accepts if
\begin{equation}
\bigabs{\{ j \in [N m] : A_j B_j = X_j \oplus Y_j \}} \geq \cos^2(\pi/8) N m - \tfrac{1}{2 \sqrt 2} \sqrt{N m \log(N m)}
 \enspace .
\end{equation}
\item
State tomography: In the second protocol, Eve chooses $K \in [N]$ uniformly at random.  She referees $(K-1) m$ CHSH games.  For the $K$th set, she referees a state tomography protocol with parameters~$q$, $n$, $m$, $\xzdeterminedset$ and $\sigma$.  She accepts if the following criteria are satisfied:

\begin{subequations} \label{e:tomographyacceptancecriteria}
\begin{align}
\max_{o \in [2^q]} \bigabs{\# \{ j : O_j = o \} - n/2^q} &\leq 4^q \sqrt{n \log n} \label{e:tomographyacceptanceenoughstatisticsperoutcome}
\\
\max_{o \in [2^q], P \in \{I, X, Z\}^{\otimes q}} \abs{\tau^{o, P} - \Tr (\pi^o P)} &\leq 4^q \sqrt{(\log n)/n} \label{e:tomographyacceptancecriterionestimator}
 \enspace .
 \end{align}\end{subequations}

\end{enumerate}

The combined protocol satisfies the following completeness and soundness conditions:
\begin{description}
\item[Completeness:]
If Alice and Bob use $N m$ shared EPR states to play the CHSH games according to an ideal strategy, and if Bob uses an ideal strategy with respect to the projections~$\xzdeterminedset$ on the~$K$th set of $m$ EPR states in the state tomography protocol, then in both protocols,
\begin{equation}
\Pr[\text{Eve accepts}] \geq 1 - O(n^{-1/2})
 \enspace .
\end{equation}
\item[Soundness:]
Assume that for both protocols, $\Pr[\text{Eve accepts}] \geq 1 - n^{-1/3}$.  Let $\rho$ be Alice's state in the second protocol after $(K-1) m$ games and conditioned on Bob's messages \\ $O_1, \ldots, O_n$.  Then there exists an isometry $\XA : \H_A \hookrightarrow (\C^2)^{\otimes m} \otimes \H_A'$ such that letting $\rho_{\sigma, j}$ be $\XA \rho \XA{}^\dagger$ reduced to Alice's qubits $\{ \sigma(j, i) : i \in [q] \}$,
\begin{equation}
\Pr\!\Big[
\bigabs{ \big\{ j \in [n] : \Tr (\rho_{\sigma, j} \pi^{O_j}) \geq 1 - O(n^{-1/16})
\big\} } \geq \big(1 - O(n^{-1/16})\big)n
\Big] \geq 1 - 4 n^{-1/12}
 \enspace .
\end{equation}
Here, the probability is over $K$, the first $(K-1) m$ games and $O_1, \ldots, O_n$.

\end{description}
\end{theorem}

We give the following corollary to this theorem:

\begin{corollary} \label{c:cor1}
By changing the measurement operators accordingly, a state tomography protocol for $q$-qubit $X\!Y$-determined ($Y\!Z$-determined) states exists,
and achieves the same completeness and soundness bound as the one from Theorem~\ref{t:statetomography}.
\end{corollary}

\begin{proof}
As mentioned, if we consider the extended CHSH game (comprising of $6$ CHSH games) and try to rigidly determine the existence of a tensor
product of $\ket{\Psi^+}$ states, we can fix the strategies of the provers up to an $X\!Y$ plane reflection. In particular, the results of
an $X\!Z$ state tomography in the original setting hold here for $X\!Y$ ($Y\!Z$) tomography. Therefore, it is possible to certify the preparation of $q$-qubit $X\!Y$-determined ($Y\!Z$-determined) states.
\end{proof}

\noindent{We can now present the main lemma proving that Protocol~\ref{prot:modifiedstatetomography} is a verification protocol.}

\begin{lemma} \label{l:modifiedstatetomography}
Protocol~\ref{prot:modifiedstatetomography} has completeness lower bounded by $1 - O(n^{-1/2})$ and soundness upper bounded by $O(n^{-1/12})$.
\end{lemma}
\begin{proof}
According to corollary~\ref{c:cor1}, the six state tomography protocols that constitute
Protocol~\ref{prot:modifiedstatetomography} are
valid verification protocols achieving the same bounds for completeness and soundness as the original protocol from Theorem~\ref{t:statetomography}.
We will ignore the case of $X\!Y$ plane reflections since, as we have shown in Lemma~\ref{l:reflectedFK}, these are detected with overwhelming probability by the FK protocol.
These protocols can be ``stitched'' together in the same way the subprotocols of the RUV protocol are stitched together. In fact, our case requires
a much simpler analysis since the six state tomography protocols are independent of each other. This means that in each subprotocol, the verifier
is not basing his questions on the results of any previous subprotocol. This nonadaptive technique contrasts the RUV protocol in which the questions were adaptive.
In the case of honest provers, the verifier accepts if all subprotocols succeed. For each one, we know
from Theorem~\ref{t:statetomography}
that the probability of acceptance is
$\geq 1 - O(n^{-1/2})$, hence for the whole protocol the probability of acceptance is $\geq (1 - O(n^{-1/2}))^6 = 1 - O(n^{-1/2})$. Thus, we see
that the completeness bound remains unchanged.
For soundness, assuming the provers are dishonest we know, again from Theorem~\ref{t:statetomography}, that the probability of accepting an
incorrect outcome is $\leq 4n^{-1/12}$. In our protocol, the provers can be dishonest in any of the six subprotocols, therefore, by a union bound
the probability of accepting an incorrect outcome is $\leq 6 \cdot 4n^{-1/12} = 24n^{-1/12}$. Therefore, we can say that the soundness of our protocol
is upper bounded by $O(n^{-1/12})$.
\end{proof}

We are now able to give the proof of our main result (Theorem~\ref{t:composite}) which concerns the properties of the composite protocol (Protocol~\ref{prot:composite}). To do this, we require an additional property.
\begin{lemma} \label{l:closeness}
Assume the verifier
wants to prepare a state $\rho$ consisting of tensor products of qubits which are all determined in either the $X\!Z$, $X\!Y$ or $Y\!Z$ bases.
A successful run of Protocol~\ref{prot:modifiedstatetomography} certifies that, 
prover $1$ has a state $\rho'$ such that
$\rho$ and $\rho'$ are close in trace distance.
\end{lemma}

\begin{proof} The proof is partially given in \cite{ruv}. In the state tomography protocol, a prover prepares multiple copies of
a resource state. In \cite{ruv} it is stated that if the verifier accepts, then, with high probability,
a subset of states of the prover are close in trace distance
to copies of the resource state.

The soundness condition of Theorem~\ref{t:statetomography} states that:
\begin{equation}
\Pr\!\Big[
\bigabs{ \big\{ j \in [n] : \Tr (\rho_{\sigma, j} \pi^{O_j}) \geq 1 - O(n^{-1/16})
\big\} } \geq \big(1 - O(n^{-1/16})\big)n
\Big] \geq 1 - 4 n^{-1/12}
 \enspace .
\end{equation}
It is shown in \cite{ruv} that this condition translates to the fact that with probability at least $1 - O(n^{-1/48})$ we have:
\begin{equation} \label{eqn:closeness}
\| \rho_{S}(O_{1,n}) - \otimes_{j \in S} \pi^{O_j} \|_{Tr} \leq O(n^{-1/64})
\end{equation}
Where $S$ is uniformly random subset of size $O(n^{1/64})$.
If we denote $O(n^{-1/64})$ as $\epsilon$, $O(n^{-1/48})$ as $p$, $\rho_{S}(O_{1,n})$ as $\rho_{\epsilon}$ and $\otimes_{j \in S} \pi^{O_j}$ as $\rho_{id}$ then the state $\rho'$, that prover $1$ has, is:
\begin{equation}
\rho' = (1 - p) \rho_{\epsilon} + p (I - \rho_{\epsilon})
\end{equation}
We can see that, for sufficiently large $n$, the values of $p$ and $\epsilon$ tend to 
$0$. Consequently, $\rho'$ approaches $\rho_{\epsilon}$ and $\rho_{\epsilon}$ approaches the ideal state, $\rho_{id}$. 
Computing the trace distance between $\rho'$ and $\rho_{id}$, we obtain: 
\begin{equation}
\| \rho' - \rho_{id} \|_{Tr} = \| (1 - p) \rho_{\epsilon} + p (I - \rho_{\epsilon}) - \rho_{id} \|_{Tr} \leq
\| \rho_{\epsilon} - \rho_{id} \|_{Tr} + \| p (I - 2\rho_{\epsilon}) \|_{Tr} 
\end{equation}
And using inequality~\ref{eqn:closeness}, we have:
\begin{equation}
\| \rho' - \rho_{id} \|_{Tr} \leq O(n^{-1/64}) + 2p = O(n^{-1/64}) + O(n^{-1/48}) = O(n^{-1/64})
\end{equation}
Therefore, the state that prover $1$ has, conditioned on his messages $O_{1,n}$, is close to the state comprised of copies of the resource states.
Depending on which type of state tomography is done, the resource states are determined in either the $XZ$, $XY$ or $YZ$ bases.
\end{proof}

\begin{proof}[Proof of \emph{\textbf{Theorem~\ref{t:composite}}}]
According to Lemmas~\ref{l:modifiedstatetomography} and~\ref{l:closeness}, Protocol~\ref{prot:modifiedstatetomography} is capable of preparing
with high probability, a multi-qubit state $\rho$ on prover $1$'s side,
such that $\rho$ is $\epsilon$-close to a tensor product of states determined
in either the $XZ$, $XY$ or $YZ$ bases. In fact, each subprotocol is capable of such a preparation.
For the two $XY$ state tomography protocols we choose the resource state to be:
\begin{equation}
\Ket{+} \otimes \Ket{+_{\pi/4}} \otimes \Ket{+_{2\pi/4}} \otimes \Ket{+_{3\pi/4}} \otimes \Ket{+_{4\pi/4}} \otimes
\Ket{+_{5\pi/4}} \otimes \Ket{+_{6\pi/4}} \otimes \Ket{+_{7\pi/4}}
\end{equation}
For the two $XZ$ state tomography protocols we choose the resource state to be:
\begin{equation}
\Ket{0} \otimes \Ket{1} \otimes \Ket{0} \otimes \Ket{1}
\end{equation}
This allows us to prepare multi-qubit states on prover $1$'s side which are close in trace distance to the FK input consisting of $XY$-plane states
and dummy qubits (the $\Ket{0}$, $\Ket{1}$ qubits). If we denote as $\rho_1$ the multi-qubit state consisting of multiple copies of the $XY$ resource state
and $\rho_2$ as the multi-qubit state consisting of multiple copies of the $XZ$ resource state, then the FK input is effectively $\rho_1 \otimes \rho_2$.
Lemma~\ref{l:closeness} shows that with high probability prover $1$
will have a state $\rho_1'$ that is $\epsilon_{prep}$-close to $\rho_1$ and a state
$\rho_2'$ that is $\epsilon_{prep}$-close to $\rho_2$, where $\epsilon_{prep} = O(n^{-1/64})$.
Therefore, $\rho_1' \otimes \rho_2'$ is $2\epsilon_{prep}$-close to $\rho_1 \otimes \rho_2$.
Moreover, in \cite{ruv} it is proven in the state tomography protocol prover $1$ is completely blind regarding his state.
Given this, and using Theorem~\ref{t:robust}, we can compose the modified state tomography protocol (Protocol~\ref{prot:modifiedstatetomography})
with the FK protocol to achieve a new blind verification protocol. The state $\rho_1' \otimes \rho_2'$ is used as input for the FK protocol, since
it is $\epsilon$-close to the ideal input, where $\epsilon = 2 \epsilon_{prep}$. \\
The bound on completeness for the new protocol can be computed from the completeness bounds of the constituent protocols. In the honest provers setting,
the verifier's acceptance probability for modified state tomography is $1 - O(n^{-1/2})$, and for FK with deviated input it is
$1 - O(\sqrt{\epsilon}) = 1 - O(n^{-1/128})$.
Multiplying these together and taking the leading order terms we find that completeness of the protocol is upper bounded by
$1 - O(n^{-1/128})$. \\
For soundness, in the dishonest setting if the verifier
would reject in either modified state tomography or FK then he would reject in the new protocol as
well. The bound on soundness for modified state tomography is $O(n^{-1/12})$ and for FK is
$\left( \frac{2}{3} \right)^{\lceil \frac{2d}{5} \rceil}$, where $d$ is the security parameter of the FK protocol that specifies
the size of the encoding for the computation graph.
From a union bound we get that the soundness for our composite approach is $\left( \frac{2}{3} \right)^{\lceil \frac{2d}{5} \rceil} +
O(n^{-1/12})$.

The last part of the proof deals with the round complexity of our composite approach. In the previous proof of Lemma~\ref{l:closeness} we mentioned that prover $1$'s state restricted to subset of
$O(n^{1/64})$ qubits is close in trace distance
to the ideal state. However we need $n$ to be sufficiently large so that this subset of qubits can encompass the entire FK input.
We know that the FK input comprises of $O(|\mathcal{C}|^2)$ qubits, where $\mathcal{C}$ is the computation the verifier wants to perform.
This means, that we need $O(|\mathcal{C}|^{128})$ qubits in total so that we can claim
that a state of $O(|\mathcal{C}|^2)$ qubits is close to its intended value.
Recall that from Thereom~\ref{t:statetomography},
the number of rounds for state tomography is $O(n^{\alpha})$, where we know from \cite{ruv} that $\alpha > 16$.
This means that the total number of rounds must be $O(|\mathcal{C}|^c)$, where $c > 128 \cdot 16 = 2048$.
If we relabel $n$ to be $|\mathcal{C}|$ then the round complexity is $O(n^c)$.
\end{proof}

\section{Proof of Fault Tolerance} \label{sect:ft}

The main result of this section is the proof of Theorem~\ref{t:ft5} that gives a fault tolerant FK protocol. We first prove Lemmas~\ref{l:ft1},~\ref{l:ft2},~\ref{l:ft3} and~\ref{l:ft4} that as stressed in Section~\ref{sect:main-ft}, highlights why we cannot use results similar to the robustness and why the simplest approaches fail. Then we proceed in the proof of Theorem~\ref{t:ft5}.

\begin{proof}[Proof of \emph{\textbf{Lemma~\ref{l:ft1}}}]
We can compute a bound on the trace distance between an arbitrary qubit $\rho_i$ and $\mathcal{E}(\rho_i)$:
\begin{equation}
\| \rho_i - \mathcal{E}(\rho_i)\|_{Tr} = \| \rho_i - (1 - p)\rho_i - (p/3)([X]+[Y]+[Z])\rho_i([X]+[Y]+[Z]) \|_{Tr}
\end{equation}
\begin{equation}
\| \rho_i - \mathcal{E}(\rho_i)\|_{Tr} = p \| \rho_i - (1/3)([X]+[Y]+[Z])\rho_i([X]+[Y]+[Z]) \|_{Tr}
\end{equation}
But we know that the trace distance is upper bounded by $1$, so:
\begin{equation}
\| \rho_i - \mathcal{E}(\rho_i) \|_{Tr} \leq p
\end{equation}
Now we compute the trace distance between $\sigma = \otimes_{i=1}^n \rho_i$ and $\sigma' = \otimes_{i=1}^n \mathcal{E}(\rho_i)$:
\begin{equation}
\| \sigma - \sigma' \|_{Tr} =
\| \otimes_{i=1}^n \rho_i - \otimes_{i=1}^n \mathcal{E}(\rho_i) \|_{Tr} \leq \sum \limits_{i=1}^n \| \rho_i - \mathcal{E}(\rho_i)\|_{Tr}
\leq np
\end{equation}
Since the trace distance is upper bounded by $1$ and since $np$ can exceed $1$ for sufficiently large $n$, we have:
\begin{equation}
\| \sigma - \sigma' \|_{Tr} \leq min(1, np)
\end{equation}
Consider $\sigma = \otimes_{i=1}^n \Ket{0}\Bra{0}$. Under the action of the depolarizing channel the trace distance between $\sigma$ and $\sigma'$ is
$n \| \Ket{0}\Bra{0} - \mathcal{E}(\Ket{0}\Bra{0}) \|_{Tr}$. However $\| \Ket{0}\Bra{0} - \mathcal{E}(\Ket{0}\Bra{0})\|_{Tr} = \sqrt{\frac{2p}{3}}$, therefore
the distance between $\sigma$ and $\sigma'$ is $n \sqrt{\frac{2p}{3}}$. For sufficiently large $n$ this can clearly reach the maximum value of $1$.
\end{proof}
\begin{proof}[Proof of \emph{\textbf{Lemma~\ref{l:ft2}}}]
Because of the action of the partially depolarizing channel, each trap qubit has a probability $p$ of being changed.
We make the simplifying
assumption that an affected qubit will produce a wrong measurement result. This assumption is only valid for completeness, where
we assume that the devices are honest but faulty\footnote{If the devices were dishonest, we would need to take into account the
deviation on the trap qubits resulting from malevolent behaviour.}.
We then have that the probability of a trap measurement producing a correct outcome
is upper bounded by $1 - p$.
Given that trap measurements
are independent of each other, and assuming we have $N_T$ traps, the probability
that all trap measurements produce correct outcomes is upper bounded by $(1 - p)^{N_T}$.
Since the verifier accepts if and only if all trap measurements succeed it follows that the completeness is upper bounded by $(1 - p)^{N_T}$.
\end{proof}
\begin{proof}[Proof of \emph{\textbf{Lemma~\ref{l:ft3}}}]
Define the following Bernoulli random variable:
\begin{equation}
X_t =\begin{cases}
1, & \text{if measurement of trap } t \text{ fails}.\\
0, & \text{otherwise}.
\end{cases}
\end{equation}
Under the simplifying assumption of the previous lemma, we have $Pr(X_t = 1) = p \geq 0$.
Next, we define:
\begin{equation}
F = \sum_{t = 1}^{N_T} X_t
\end{equation}
It is clear that $E(F) = N_T p$. Additionally, using
a Hoeffding inequality, we have that:
\begin{equation}
Pr(F \geq (p + \epsilon) N_T) \leq \exp(-2\epsilon^2N_T)
\end{equation}
This gives the probability that the number of failed traps is greater than our threshold of $pN_T$. The complement of this is
the completeness, which is therefore bounded by $1 - \exp(-2\epsilon^2N_T)$
\end{proof}
\begin{proof}[Proof of \emph{\textbf{Lemma~\ref{l:ft4}}}]
Recall that soundness is the probability of accepting an incorrect outcome. In the original FK protocol
this meant that all the traps succeeded but the computation output was orthogonal to the correct output. This is expressed with the
projector $P_\perp \otimes P_{T}$. Here $P_\perp$ projects the computation output onto the orthogonal state and $P_{T}$ projects
the trap outputs onto the correct outputs. If the accepting condition is given by a threshold of correct traps, the projector must change
accordingly. This means that there should not be only one trap projector but one for each accepting situation.
Taking the threshold to be $p N_T$ means that the verifier accepts if $T = N_T - pN_T$ traps succeeded. Since these traps can be any combination
of $T$ out of the possible $N_T$, there are $N_T \choose T$ possible accepting situations.
Therefore, the trap projector $P_{T}$ becomes a sum of $N_T \choose T$ projectors (one for each accepting choice of traps).
It therefore follows from linearity that the soundness bound becomes ${N_T \choose T}
\left( \frac{2}{3} \right)^{\lceil \frac{2d}{5} \rceil}$.
\end{proof}

\begin{proof}[Proof of \emph{\textbf{Theorem~\ref{t:ft5}}}]
In \cite{topo}, Morimae and Fujii show how a blind quantum computation can be made fault tolerant by encoding it in a topologically protected
error-correcting code \cite{rhg}. The encoding then uses a decoration trick so that the prover only needs to perform $XY$-plane measurements
and this can be done blindly using the UBQC protocol in \cite{fk}.
Here, we use the same idea to encode a computation which also contains an isolated trap. This follows from the first verification protocol
introduced in \cite{fk} which uses a brickwork state to perform the computation. Encoding this in the fault tolerant code give us the lattice
$\mathcal{L}^{\nu}$, which according to \cite{topo} can be executed blindly by the prover. \\
Throughout the run of the protocol, if the prover is always honest then the fault tolerant code will correct for any errors (since we have
assumed the error rate is smaller than the threshold of correctable errors). This proves that the completeness of the protocol is $1$. \\
To compute the soundness, note that the computed bound for the brickwork state protocol in \cite{fk} is:
\begin{equation}
p_{incorrect} < \left( 1 - \frac{1}{2n} \right)
\end{equation}
Where $n$ is the number of qubits in the brickwork state.
Similar to the
robustness proof, the proof of this bound
assumes that the outcome density operator of the protocol is projected onto a state where the trap succeeded but the
computation outcome is incorrect.
It is can been shown that the same bound as the non-fault tolerant case holds.
This means that we have:
\begin{equation}
p_{incorrect} < \left( 1 - \frac{1}{2n'} \right)
\end{equation}
Where $n'$ is the number of qubits in a lattice $\mathcal{L}^{\nu_i}$,
out of the $N$ lattices used in the protocol. We note that $n'$ is of the same order as $n$ \cite{rhg}, and we can choose a constant $c>2$ such that $2n'=cn$. \\
In Protocol~\ref{prot:ft} the verifier creates \emph{independent} encodings $\mathcal{L}^{\nu}$, each depending on classical randomness. He accepts
the sequence of encodings if all trap measurements succeed in each encoding.
This means that the prover
can deceive the verifier if he can deviate the computation in each encoding $\mathcal{L}^{\nu}$ while at the same time passing all the
traps. However, for any given encoding we know that the probability of this happening is given by $p_{incorrect}$,
and because of independence, the prover will succeed for the
sequence with probability:
\begin{equation}
p^N_{incorrect} < \left( 1 - \frac{1}{cn} \right)^N
\end{equation}
We know that $N/R > 1/\log \left( \frac{cn}{cn - 1} \right)$.
However, this is equivalent to:
\begin{equation}
N/R > -1/\log\left( 1 - \frac{1}{cn} \right)
\end{equation}
\begin{equation}
(N/R) \log \left( 1 - \frac{1}{cn} \right) < -1
\end{equation}
Note that we used the fact that $\log \left( 1 - \frac{1}{cn} \right) < 0$.
Through exponentiation we get:
\begin{equation}
\left( 1 - \frac{1}{cn} \right)^{N/R} < \frac{1}{2}
\end{equation}
And we finally obtain:

\begin{equation}
p^N_{incorrect} < \frac{1}{2^R}
\end{equation}
Hence, the probability that the prover deceives the verifier is less than $(1/2)^R$ and so the soundness of the protocol is
upper bounded by this value. \\
Lastly we compute the round complexity of this protocol. For the given sequence we have $N$ encodings and for each encoding we have $O(n)$ qubits\footnote{Note
that here, unlike in the composite protocol, we have a linear number of qubits for the protocol. This is because we are not using the dotted-complete version of the
FK protocol, but the one using a brickwork state. It is explained in \cite{fk} that this version of the protocol requires $O(n)$ qubits.}
and a corresponding round complexity of $O(n)$ to compute the execution of that encoding. It follows that the overall complexity is $O(Nn)$.
But we know that $N < R/ \log \left( 1+\frac{1}{cn-1} \right) + O(1)$, and
given that $R$ is a constant,
we can show that $N$ is $O(n)$. This follows from the observation that dividing the function $1/\log\left(1+\frac{1}{cn-1}\right)$  with $(cn-1)$ gives a constant in the $n\rightarrow\infty$ limit:

\begin{equation}
\lim_{n\rightarrow\infty}(cn-1)\log\left(1+\frac{1}{cn-1}\right)=\frac1{\ln2}
\end{equation}
Incorporating this result yields overall complexity
$O(n^2)$.
Note that this proof technique works for the case of classical output since we are interested in the classical output of each encoding.
The encodings are independent from each other, which allows us to bound the probability for the whole sequence.
\end{proof}

\section{Conclusion}\label{sect:conclusion}

We have shown that the single server universal verifiable blind quantum computing protocol of \cite{fk} is robust even against general adversaries. This protocol is currently the 
optimal protocol in terms of the verifier's requirements. The robustness result further strengthens the scheme for realistic applications where the effect of noisy devices should also be considered, as highlighted in a recent experimental demonstration of the protocol \cite{efk}. Moreover, it enables us to compose the FK protocol with other quantum verification protocols, extend it to the entangled servers setting and make it device independent. The key property that we proved, is that the protocol remains secure even against correlated attacks. To achieve this, we considered the deviation of the evolution of a correlated subsystem from the evolution of uncorrelated subsystems. The former could be written mathematically as a non-CPTP map which differs from a CPTP map by an inhomogeneous term. However, for inputs which are $\epsilon$-close to the ideal FK input, we showed that this deviation (the inhomogeneous term) is bounded by a term of order $O(\sqrt{\epsilon})$. Our proof technique is generic and can be potentially applied to other multi-party protocols where sequential composition is required. This result complements the \emph{local}-verifiability proof of \cite{DFPR13} which is based on the universal composability framework. The latter, in its current form, is insufficient for composing entanglement-based protocols, such as RUV, with the FK protocol because of the possibility of correlated attacks. Our robustness result, however, leads to a stand alone secure composite verification protocol. Additionaly, the proposed composition scheme could potentially be used to extend the composable framework of \cite{DFPR13} to incorporate multiple provers.

Our proposed composite protocol achieves verification with a classical client (device independence) and gives improved round complexity in comparison to the RUV protocol. It uses only the (modified) state tomography part of RUV as input for the FK protocol.
The improved round complexity of the composite protocol is still too high to allow for any practical implementation
in the near future. However, the reason for this high round complexity is the state tomography subprotocol and therefore, any improvement on how to prepare the FK inputs (e.g. by exploiting the shared entanglement of the provers or using self-testing techniques as in \cite{joe}) will directly improve the efficiency of our composite protocol as well.

Finally we outlined how to make our verification protocol fault tolerant. To do so we constructed a
fault tolerant version of the FK protocol which is interesting in its own right. This complements the work presented in \cite{topo} which addresses the fault tolerance of a (non-verifiable) blind quantum computing protocol. We used the same topological error correcting code as \cite{topo} and
a sequential repetition scheme in order to correct for faulty devices.

\section*{Acknowledgements}

Shortly before uploading a prepint on the arxiv, the authors became aware of parallel and independent research by Hajdusek, Perez-Delgado and Fitzsimons, which also addresses device-independent verifiable blind quantum computing and appeared the same day in the arxiv \cite{joe}. We would like to thank Vedran Dunjko and Theodoros Kapourniotis for useful discussions. PW gratefully acknowledges partial support from COST Action MP1006.

\bibliography{report}
\bibliographystyle{unsrt}

\end{document}